\documentclass[11pt,reqno]{amsart}

\usepackage[hmargin=3cm,vmargin=3cm]{geometry}
\usepackage{amsmath,amssymb,comment}
\usepackage{graphicx}
\usepackage{mathtools}
\usepackage{thmtools,thm-restate}
\usepackage{color}

\usepackage{hyperref}

\usepackage{enumitem}

\newtheorem{theorem}{Theorem}[section]
\newtheorem{proposition}[theorem]{Proposition}
\newtheorem{corollary}[theorem]{Corollary}

\newtheorem{lemma}[theorem]{Lemma}

\theoremstyle{definition}

\newtheorem*{definition*}{Definition}

\theoremstyle{remark}
\newtheorem{remark}[theorem]{Remark}

\numberwithin{equation}{section}

\def\CC{\mathbb{C}}
\def\NN{\mathbb{N}}
\def\RR{\mathbb{R}}

\def\ZZ{\mathbb{Z}}

\newcommand{\cA}{{\mathcal A}}
\newcommand{\cB}{{\mathcal B}}
\newcommand{\cC}{{\mathcal C}}

\newcommand{\cE}{{\mathcal E}}

\newcommand{\cH}{{\mathcal H}}
\newcommand{\cI}{{\mathcal I}}

\newcommand{\cL}{{\mathcal L}}

\newcommand{\cO}{{\mathcal O}}

\newcommand{\cR}{{\mathcal R}}

\newcommand{\cX}{{\mathcal X}}

\newcommand{\supp}{\operatorname{supp}}

\DeclareMathOperator\Div{div}

\DeclareMathOperator\curl{curl}

\DeclareMathOperator\Int{int}

\author{Daniel Peralta-Salas}
\address{Instituto de Ciencias Matem\'aticas, Consejo Superior de Investigaciones Cient\'ificas, 28049 Madrid, Spain \newline \it{e-mail: dperalta@icmat.es}}

\author{David Perrella}
\address{Instituto de Ciencias Matem\'aticas, Consejo Superior de Investigaciones Cient\'ificas, 28049 Madrid, Spain \newline \it{e-mail: david.perrella@icmat.es}}

\author{David Pfefferl\'e}
\address{The University of Western Australia, 35 Stirling Highway, Crawley WA 6009, Australia \newline \it{e-mail: david.pfefferle@uwa.edu.au}}

\thanks{We thank Gunnar Hornig who during his visit to the ICMAT inspired this research, Jason Cantarella for sharing unpublished work which informed the proof of \eqref{eq:jm-11-bound-for-kappa} in Lemma \ref{lem:main-bounding-work}, and Enrico Valdinoci and Serena Dipierro for interesting discussions. We also thank Naoki Sato for valuable feedback on the paper. The authors are partially supported by the grants CEX2023-001347-S, RED2022-134301-T and PID2022-136795NB-I00 (D.P.-S. and D.Pe.) funded by MCIN/AEI/10.13039/501100011033, and Ayudas Fundaci\'on BBVA a Proyectos de Investigaci\'on Cient\'ifica 2021 (D.P.-S.).}

\begin{document}

\begin{abstract}
We construct families of rotationally symmetric toroidal domains in $\mathbb R^3$ for which the eigenfields associated to the first (positive) Amp\`erian curl eigenvalue are symmetric, and others for which no first eigenfield is symmetric. This implies, in particular, that minimizers of the celebrated Woltjer's variational principle do not need to inherit the rotational symmetry of the domain. This disproves the folk wisdom that the eigenfields corresponding to the lowest curl eigenvalue must be symmetric if the domain is. 
\end{abstract}

\title{Asymmetry of curl eigenfields solving Woltjer's variational problem}

\maketitle


\section{Introduction}

Let $\Omega \subset \RR^3$ be a bounded domain with smooth boundary. If $\cX(\Omega)$ denotes the space of $C^\infty$ vector fields (up to the boundary) on $\Omega$, then a (smooth) magnetic field in $\Omega$ is an element of
\begin{equation*}
\cB(\Omega) = \{B \in \cX(\Omega) : \Div B = 0,\, B \cdot n = 0\},
\end{equation*}
where $n$ denotes the outward unit normal on $\partial \Omega$.

A widely studied subject in the mathematical theory of plasmas is the so-called Woltjer's variational principle~\cite{W58,LA91,AK21,gerner2020existence}, which consists in the minimisation problem 
\begin{equation}\label{eq:min-prob-euclidean}
\cE(B) \rightarrowtail \text{min},\qquad \cH(B) = h
\end{equation}
for $h \neq 0$, where $\cE = \|\cdot\|_{L^2}^2$ denotes the $L^2$-energy, and $\cH$ is the helicity~\cite{M69,berger-field-1984,finn-antonsen-1985}, which is given by the well known formula
\begin{equation*}
\cH(B) = \frac{1}{4\pi} \int_{\Omega \times \Omega} B(x) \times B(y) \cdot \frac{x-y}{|x-y|^3} d\mu(x)d\mu(y).
\end{equation*}
This minimisation problem was studied in \cite{LA91,cantarella2001upper} and the following was shown. We shall restrict to the case $h>0$, the case $h<0$ being completely analogous.

\begin{theorem}\label{thm:Cantarella-existence-R3}
For any $h >0$, Problem \eqref{eq:min-prob-euclidean} admits a solution over $\cB(\Omega)$. Every solution $B \in \cB(\Omega)$ satisfies, for some $\lambda>0$, that
\begin{equation}\label{eq:Beltrami}
\curl B = \lambda B, \qquad B \cdot n = 0,    
\end{equation}
together with the additional boundary condition that
\begin{equation}\label{eq:amp-condition}
\int_{\partial D} B \cdot dl = 0,    
\end{equation}
for all embedded closed disks with $D \setminus \partial D \subset \RR^3\setminus \Omega$. Moreover, $\lambda$ is the smallest positive number for which the spectral problem given by~\eqref{eq:Beltrami} and~\eqref{eq:amp-condition} holds.
\end{theorem}

In view of Theorem~\ref{thm:Cantarella-existence-R3}, the minimisation problem~\eqref{eq:min-prob-euclidean} leads to the study of the first (positive) curl eigenvalue on a domain with suitable boundary conditions. To state the spectral problem precisely, we need to introduce some notation.
Letting $\cB_\cA(\Omega)$ denote the set of \emph{Amp\`erian knots}, that is the set of $B \in \cB(\Omega)$ which satisfy Equation \eqref{eq:amp-condition}, known as the \emph{Amp\`erian boundary condition}, then the curl restricts as an operator
\begin{equation*}
\curl_{\cA} : \cB_\cA(\Omega) \to \cB(\Omega)
\end{equation*}
which has a compact self-adjoint inverse $\curl^{-1}_\cA : L^2\cB(\Omega) \to L^2\cB(\Omega)$ defined on the $L^2$-closure of $\cB(\Omega)$. We call the eigenfunctions of $\curl_{\cA}$ \emph{Amp\`erian curl eigenfields} and we call the eigenvalues of $\curl_{\cA}$ \emph{Amp\`erian curl eigenvalues}. The distinct Amp\`erian curl eigenvalues may be written as
\begin{equation*}
\cdots < \lambda_{-2} < \lambda_{-1} < 0 < \lambda_1 < \lambda_2 < \cdots,    
\end{equation*}
with $\lambda_n \to \pm \infty$ as $n \to \pm \infty$. If $B$ is an Amp\`erian curl eigenfield with eigenvalue $\lambda \neq 0$, then
\begin{equation*}
\frac{\cH(B)}{\cE(B)} = \frac{1}{\lambda}.
\end{equation*}
Thus, Theorem \ref{thm:Cantarella-existence-R3} states that the minimisers of Problem \eqref{eq:min-prob-euclidean} are the suitably scaled eigenfields corresponding to the first (positive) Amp\`erian curl eigenvalue.

Our work concerns the following natural question about the first Amp\`erian curl eigenfields:
\begin{quote}
\emph{In the case that the domain $\Omega$ is rotationally symmetric, do the first Amp\`erian curl eigenfields inherit the rotational symmetry?}
\end{quote}
This question is of course a classical one of spectral geometry. It is typically resolved \emph{in the positive} as a by-product of stronger results which hold regardless of the domain. As well-known \cite{Bhattacharya88,Martinez02,Abatangelo24}, in the case of functions, symmetry can be deduced from simplicity of the first eigenvalue. This is what happens for both the Neumann and the Dirichlet Laplacian for example (with the first eigenfunctions being constants in the Neumann case).

Although there is not much literature explicitly dedicated to this question for the curl operator, there are still some results. In the case when $\Omega$ is simply connected, the Amp\`erian boundary condition is redundant and the problem is just about eigenvalues of the curl operator. In the affirmative direction it was proven in \cite{Cantarella_DeTurck_Gluck_Teytel_2000} that when $\Omega$ is a ball, or when $\Omega$ is the region bounded by two concentric spheres, then the first curl eigenfields each inherit a rotational symmetry from the domain, despite there being multiplicity. In the negative, it has been numerically observed \cite{Bondeson81,Finn81} that on solid cylinders of finite height the first curl eigenfields need not inherit the rotational symmetry. 

The story is different when $\Omega$ is a toroidal domain. Here, the distinction between the Amp\'erian curl spectrum and the flux-free curl spectrum~\cite{GY90} matters. For the Amp\`erian spectrum, the answers have only been in the \emph{afirmative}. Numerical studies \cite{Tsuji_1991,Gimblett_Hall_Taylor_Turner_1987} have observed that the first Amp\`erian curl eigenfields are rotationally symmetric in a variety of cases (called the ``zero-field eigenvalues" or ``field reversal-point" in these studies), including the standard solid torus with arbitrary major and minor radii, and idealised cases where the domain is a flat periodic cylinder (i.e., flat tube) with sections including a rectangle of arbitrary length (rotation here is taken to mean periodic translation). It was also stated in \cite{cantarella2001upper} for the flat tube with a disk section that the first Amp\`erian curl eigenfields inherit every symmetry of the domain. For the flux-free curl spectrum however, it has been observed numerically that the first curl eigenfields need not inherit the symmetry in some toroidal domains \cite{Faber_Jr_Wing_1985} in $\RR^3$ and the aforementioned flat tube with disk cross-section \cite{Yoshida91}. 

A tempting explanation for the observed symmetry of first Amp\`erian curl eigenfields on rotationally symmetric toroidal domains, despite symmetry breaking in first flux-free curl eigenfields, lies in the fact that rotationally symmetric curl eigenfields in solid tori are associated with eigenfunctions of the Grad-Shafranov operator (c.f. propositions \ref{prop:amp-sym-is-G-S} and \ref{prop:flux-free-sym-is-G-S} in sections \ref{sec:sym-amp-intro} and \ref{prop:flux-free-sym-is-G-S}), with the Dirichlet boundary condition (exploited heavily in \cite{cantarellathesis}) in the Amp\`erian case and a zero average condition in the flux-free case, representing orthogonality to the unique (up to scaling) curl-free magnetic field on $\Omega$ (known as the vacuum field), which is rotationally symmetric. The idea being that the first flux-free eigenvalue is actually a ``second" curl eigenvalue of a larger spectrum, with the ``true" first eigenvalue being $0$, corresponding to the symmetric vacuum field. This is in analogy with first Neumann eigenvalue being $0$ with all eigenfunctions being constant (therefore inheriting symmetry) and that one does not see symmetry inherited for the first non-zero Neumann eigenvalue. Whereas in the Amp\`erian case, the Dirichlet boundary condition means that the first symmetric Amp\`erian eigenfield ``really is the first eigenvalue" and so inherits symmetry just like the first Dirichlet eigenfunction of the Laplacian. This analogy is not as direct in the simply connected case because there the Grad-Shafranov equation becomes singular. 

In view of the aforementioned numerical studies and analogies with Laplacian eigenvalues, the folk wisdom dictates that on solid tori, the answer should be ``yes". As it happens however, we prove that the answer is \emph{no} in general. We will see that the answer depends on geometry. To focus ideas, for an embedded closed disk $\Sigma \subset \RR^+ \times \RR$, we call
\begin{equation*}
\cR_\Sigma = \{(x,y,z) \in \RR^3 : (\sqrt{x^2+y^2},z) \in \Sigma\}    
\end{equation*}
the \emph{smooth solid torus of revolution with cross-section $\Sigma$}. For instance, for $0 < a < R$, consider the standard solid torus of minor radius $a$ and major radius $R$ is given by $\cR_{a,R} = \cR_{D_{a,R}}$ where $D_{a,R}$ is the closed disk of radius $a$ centred at $(R,0)$. We show that the question is resolved positively for small enough minor radius.

\begin{theorem}\label{thm:sym-wins}
For any $R>0$, and any $a>0$ sufficiently small, the first positive Amp\`erian curl eigenvalue on the solid torus of revolution $\cR_{a,R}$ is simple and the corresponding curl eigenfield is rotationally symmetric.
\end{theorem}

This is in agreement with the aforementioned \cite{Tsuji_1991}. In the negative, we have the following.

\begin{theorem}\label{thm:antisym-wins}
For $0 < a < b$ and $L > 0$, set
\begin{align*}
A_{a,b} &= \{(x,y) \in \RR^2 : a^2 \leq x^2+y^2 \leq b^2\},\\
\cC_{a,b,L} &= A_{a,b} \times [-L/2,L/2].
\end{align*}
Fix $b > 0$ and set $L_{n} \coloneqq n L$, $n \in \NN$. There exists $0 < a < b$, $L > 0$ and $N \in \NN$ such that, for any  smooth solid torus of revolution $\cR \subset \RR^3$ satisfying
\begin{gather*}
\cC_{a,b,L_{n}} \subset \cR \subset \cC_{a,b,L_{n+1}},\\
\partial A_{a,b} \times [-L_{n}/2,L_{n}/2] \subset \partial \cR,
\end{gather*}
with $n\geq N$, there do not exist rotationally symmetric curl eigenfields corresponding to the first (positive) Amp\`erian curl eigenvalue on $\cR$.
\end{theorem}

\begin{remark}\label{rem:thms}
Some comments on the theorems are in order:
\begin{enumerate}
    \item The analogous statement in theorems \ref{thm:sym-wins} and \ref{thm:antisym-wins} for the first negative Amp\`erian curl eigenvalue also holds for the same domains thanks to the orientation-reversing isometry $(x,y,z) \mapsto (-x,y,z)$ of $\RR^3$.
    \item The curl eigenfields such as those in Theorem \ref{thm:sym-wins} were shown in \cite{cantarellathesis} to have no zeros. On the other hand, the curl eigenfields in Theorem \ref{thm:antisym-wins} are gradient fields on the boundary (see Corollary \ref{cor:grad-flows}) and thus have zeros.
    \item If $\cR$ is a smooth solid torus of revolution for which there do not exist rotationally symmetric eigenfields corresponding to the first (positive) Amp\`erian curl eigenvalue, then there do not exist rotationally symmetric eigenfields corresponding to the first (positive) flux-free curl eigenvalue either. In particular, Theorem \ref{thm:antisym-wins} is stronger than the conclusions drawn from the numerics in~\cite{Faber_Jr_Wing_1985}. See Appendix \ref{app:proof-of-amp-to-flux-free}.\label{rem:amp-to-flux-free}
\end{enumerate}
\end{remark}

The proof of both theorems \ref{thm:sym-wins} and \ref{thm:antisym-wins} involves analysing the Amp\`erian eigenvalues associated with symmetric curl eigenfields and Amp\`erian eigenvalues associated with non-symmetric curl eigenfields. Both kinds of eigenvalues come from a spectrum. In fact, we find that the non-symmetric curl eigenfields are captured by what we call the \emph{antisymmetric curl spectrum}, and in a strong sense this spectrum exists independently of the Amp\`erian curl spectrum. The antisymmetric curl spectrum can be understood as a rigorous account of the ``periodic modes" often encountered in the physics literature (c.f. the aforementioned numerical studies). We essentially play the concrete nature of the symmetric Amp\`erian spectrum, identified with the Dirichlet spectrum of the Grad-Shafranov operator, against the abstract nature of the antisymmetric curl spectrum, finding estimates for the lowest antisymmetric curl eigenvalues by making comparisons in Riemannian 3-manifolds with flat boundaries. For both theorems, the lowest Dirichlet eigenvalues of the Grad-Shafranov operator are concretely estimated thanks to the operator's close relationship with the Laplacian operator. 

For Theorem \ref{thm:sym-wins}, we compare the geometry of $\cR_{a,R}$ to that of flat tubes $D_a \times \RR/L\ZZ$ and $D_a$ the disk of radius $a$ centred at the origin to compare the corresponding antisymmetric curl-eigenvalues. For Theorem \ref{thm:antisym-wins}, we compare the geometry of $\cR$ to that of flat tubes $A_{a,b} \times \RR/L\ZZ$. A difference in Theorem \ref{thm:antisym-wins} is that we must perform a ``cut" because $A_{a,b} \times \RR/L\ZZ$ is not topologically a solid torus. We instead think of the $\RR/L\ZZ$ factor as ``vertical" and consider the fundamental domain $A_{a,b} \times [-L/2,L/2] \subset \RR^3$, which is topologically a solid torus and estimate the changes to the curl eigenvalues in doing this.

The remainder of the paper is structured as follows. In Section \ref{sec:prelim}, we formally introduce the symmetric Amp\`erian curl spectrum and the antisymmetric curl spectrum and state a key lemma helpful for proving theorems \ref{thm:sym-wins} and \ref{thm:antisym-wins}. In Section \ref{sec:proof-of-1st-thm} we prove Theorem \ref{thm:sym-wins}. In Section \ref{sec:proof-of-2nd-thm} we prove Theorem \ref{thm:antisym-wins}. In Section \ref{sec:antisymm-ineq-warped-prod} we prove an antisymmetric curl eigenvalue estimate needed for Theorem \ref{thm:sym-wins}. In Section \ref{sec:complex-antisym-fields} we introduce some complex algebra for more compactly handling the antisymmetric curl spectrum in service of proving the remaining lemmas used in theorems \ref{thm:sym-wins} and \ref{thm:antisym-wins}. In Section \ref{sec:analysis-disk-anti-sym} we prove an inequality for the antisymmetric curl eigenvalues in the aforementioned flat tube $D_a \times \RR/L\ZZ$ for Theorem \ref{thm:sym-wins}. In Section \ref{sec:upper-bound-antisym-rectangular-section} we estimate the change in antisymmetric curl eigenvalue by performing the aforementioned ``cut" on $A_{a,b} \times \RR/L\ZZ$ for Theorem \ref{thm:antisym-wins}. Lastly in Section \ref{sec:analysis-annulus-anti-sym}, we prove the existence of small antisymmetric curl eigenvalues for suitable parameters on the tubes $A_{a,b} \times \RR/L\ZZ$ for Theorem~\ref{thm:antisym-wins}.

\section{Preliminaries}\label{sec:prelim}

The goal of this section is to outline some preliminaries for the proofs of Theorems~\ref{thm:sym-wins} and~\ref{thm:antisym-wins}. Section~\ref{sec:sym-amp-intro} discusses the symmetric spectrum of the curl operator on domains that are rotationally symmetric. In contrast, Section~\ref{sec:sym-amp-intro} introduces the antisymmetric spectrum of the curl. Lastly, in Section \ref{sec:method-of-proof} we state a trivial, albeit important lemma, which is crucial in the proofs of the theorems.

\subsection{The symmetric Amp\`erian curl spectrum}\label{sec:sym-amp-intro}

Let $\cR$ be a smooth solid torus of revolution with cross-section $\Sigma \subset \RR^+ \times \RR$. Let $(x,y,z) \in \cR$ denote the Cartesian coordinates and consider the rotation vector field
\begin{equation}\label{eq:Killing-in-R3}
K = -y\partial_x + x\partial_y.
\end{equation}
As usual, an Amp\`erian curl eigenfield $B$ on $\cR$ is \emph{symmetric} if $\cL_K B = 0$ and we will call its corresponding eigenvalue a \emph{symmetric Amp\`erian curl eigenvalue}. The following standard result on the symmetric Amp\`erian spectrum of curl is proven in \cite[Proposition 3]{cantarellathesis}. 

\begin{proposition}\label{prop:amp-sym-is-G-S}
A vector field $B$ in $\cR$ is a symmetric Amp\`erian curl eigenfield if and only if there exists $u \in C^\infty(\Sigma)$ such that
\begin{align*}
B &= \frac{1}{\lambda}\nabla u(r,z) \times \nabla \varphi + u(r,z) \nabla \varphi,\\
L u &= -\lambda^2 u,\\
u|_{\partial \Sigma} &= 0,
\end{align*}
where $L$ is the Grad-Shafranov operator 
\begin{equation*}
L = \partial_{rr} + \partial_{zz} - \frac{1}{r}\partial_r,
\end{equation*}
and $(r,\varphi,z)$ denote the cylindrical coordinates. In particular, if $\lambda^{\text{G-S}}_1(\Sigma)$ denotes the first Dirichlet eigenvalue of $L$ on $\Sigma$ and $\lambda^{\text{s}}_1(\cR)$ denotes the first positive symmetric Amp\`erian curl eigenvalue, then
\begin{equation*}
\lambda^{\text{s}}_1(\cR) = \sqrt{\lambda^{\text{G-S}}_1(\Sigma)}. 
\end{equation*}
\end{proposition}

The following elementary inequalities for the Grad-Shafranov eigenvalues will be instrumental in the proof of Theorems~\ref{thm:sym-wins} and~\ref{thm:antisym-wins}. They are proved in Appendix \ref{app:Grad-Shafranov}. In the statement we endow the half-plane $\RR^+ \times \RR$ with coordinates $(r,z)$.

\begin{lemma}\label{lem:basic-G-S-ineq}
Let $U\subset U' \subset \RR^2$ be bounded domains whose closures are contained in $\RR^+ \times \RR$. Then,
\begin{equation*}
\lambda^{\text{G-S}}_1(U) \geq \lambda^{\text{G-S}}_1(U').
\end{equation*}
Moreover, letting $\lambda^D_1(U)$ denote the first Dirichlet eigenvalue of the Laplacian $\partial_{rr}+\partial_{zz}$, we have
\begin{equation*}
\lambda_1^D(U) + \frac{3}{4 r_\text{max}^2} \leq \lambda^{\text{G-S}}_1(U) \leq \lambda_1^D(U) + \frac{3}{4 r_\text{min}^2},
\end{equation*}
where 
\begin{equation*}
r_\text{min} = \inf_{(r,z) \in {U}} r,\qquad r_\text{max} = \sup_{(r,z) \in {U}} r.
\end{equation*}
\end{lemma}

\subsection{The antisymmetric curl spectrum}\label{sec:intro-antisym-spec}
Since we will use it later, it is convenient to introduce the antisymmetric spectrum in a general Riemannian setting.
Let $(M,g)$ be a compact oriented Riemannian 3-manifold with boundary. We set $\cX(M)$ to be the set of smooth vector fields on $M$ and $n$ the outward unit normal on $\partial M$. Assume that $M$ admits a (non-trivial) vector field $K$ tangent to the boundary which is Killing, i.e., $\cL_K g = 0$. We let
\begin{align*}
\cB(M) &\coloneqq \{B \in \cX(M) : \Div B = 0,\, B \cdot n = 0\},\\
\cB_K(M) &\coloneqq \cL_K(\cB(M)).
\end{align*}
A vector field $B$ is said to be a \emph{$K$-antisymmetric curl eigenfield} with corresponding \emph{$K$-antisymmetric curl eigenvalue} $\lambda$ if $B \in \cB_K(M)$ and $\curl B = \lambda B$. The set of $K$-antisymmetric curl eigenvalues is discrete and may be written as
\begin{equation}\label{eq:minimal-antisym-fact}
\cdots < \lambda^K_{-2}(M) < \lambda^K_{-1}(M) < 0 < \lambda^K_1(M) < \lambda^K_2(M) < \cdots. 
\end{equation}
See Corollary~\ref{cor:antisym-spectral} in Appendix~\ref{app:func-analysis-antisym-spec}.

The following decomposition of the spectrum of the curl operator into its symmetric and antisymmetric part is elementary; for the sake of completeness we include a proof in Appendix~\ref{app:func-analysis-antisym-spec}.

\begin{proposition}\label{prop:decomp}
Let $\lambda \in \RR \setminus \{0\}$ and set
\begin{equation*}
E_\lambda(M) = \{B \in \cX(M) : \curl B = \lambda B,\, B \cdot n = 0\}.    
\end{equation*}
Then, $E_{\lambda}(M)$ is finite-dimensional and has the $L^2$-orthogonal decomposition
\begin{equation*}
E_{\lambda}(M) = E_\lambda(M) \cap (\ker \cL_K) \oplus E_\lambda(M) \cap \cB_K(M),
\end{equation*}
The second factor further decomposes as
\begin{equation*}
E_\lambda(M) \cap \cB_K(M) = \bigoplus_{i = 1}^n E_{\lambda,K,k_i}(M), 
\end{equation*}
for some $n \in \NN$ and numbers $k_1,...,k_n > 0$, where
\begin{equation*}
E_{\lambda,K,k}(M) \coloneqq \{B \in E_{\lambda}(M) : \cL_K^2 B = -k^2 B\}.     
\end{equation*}
\end{proposition}

Now, as in Section~\ref{sec:sym-amp-intro}, we consider the case of a solid torus of revolution $\cR \subset \RR^3$. There is a Killing field $K$ given by the rotation, cf. Equation~\eqref{eq:Killing-in-R3}. In this case, we abbreviate the qualifier ``$K$-antisymmetric" to simply ``antisymmetric" and write the antisymmetric curl eigenvalues as 
\begin{equation}\label{eq:minimal-antisym-fact-R-case}
\cdots < \lambda^*_{-2}(\cR) < \lambda^*_{-1}(\cR) < 0 < \lambda^*_1(\cR) < \lambda^*_2(\cR) < \cdots. 
\end{equation}
For $\lambda \in \RR \setminus \{0\}$, we set
\begin{align*}
A_\lambda(\cR) &= \{B \in E_\lambda(\cR) : B \text{ satisfies the Amp\`erian boundary condition}\},\\
A_\lambda^{\text{s}}(\cR) &= \{B \in A_\lambda(\cR) : \cL_K B = 0\}. 
\end{align*}
By Corollary \ref{cor:grad-flows} it is obvious that $E_\lambda(\cR) \cap \cB_K(\cR) \subset A_\lambda(\cR)$, and so, by Proposition \ref{prop:decomp}, one has the $L^2$-orthogonal decomposition
\begin{equation}\label{eq:Amperian-decomp-on-R}
A_\lambda(\cR) = A_\lambda^{\text{s}}(\cR) \oplus (E_\lambda(\cR) \cap \cB_K(\cR)).
\end{equation}
Moreover, one has the following dynamical property of antisymmetric curl eigenfields. 

\begin{corollary}\label{cor:grad-flows}
Let $B$ be a antisymmetric curl eigenfield on $\mathcal R$. Then, it restricts to a gradient field on $\partial \mathcal R$.
\end{corollary}
\begin{proof}
Since $B$ is tangent to the boundary, it is well known that its metric-dual $1$-form when pulled back on $\partial\mathcal R$ is closed~\cite{EPS23}. The result then easily follows from the fact that $B$ satisfies both the Amp\`erian and the flux-free boundary conditions (cf. Lemma~\ref{lem:compatibility-of-antisym}).     
\end{proof}

\subsection{A key lemma}\label{sec:method-of-proof}

An elementary but crucial lemma in the proofs of Theorems~\ref{thm:sym-wins} and~\ref{thm:antisym-wins} is the following.

\begin{lemma}\label{lem:governing-principle}
Let $\cR \subset \RR^3$ be a solid torus of revolution. Then:
\begin{enumerate}
    \item If $\lambda^{\text{s}}_1 < \lambda^*_1$, the first (positive) Amp\`erian curl eigenvalue is simple and its corresponding eigenfield is rotationally symmetric. \label{item:sym-conclusion}
    \item If $\lambda^{\text{s}}_1 > \lambda^*_1$, there exist no rotationally symmetric eigenfields corresponding to the first (positive) Amp\`erian curl eigenvalue. \label{item:antisym-conclusion}
\end{enumerate}
\end{lemma}
\begin{proof}
Equation~\eqref{eq:Amperian-decomp-on-R} implies that the set of Amp\`erian curl eigenvalues is the union of the sets of symmetric and antisymmetric eigenvalues. 
In particular,
\begin{equation}\label{eq:governing-for-smallest-evalue}
\lambda_1 = \min\{\lambda_1^{\text{s}},\lambda_1^*\}.
\end{equation}

If $\lambda_1^{\text{s}} < \lambda_1^*$, then from Equations~\eqref{eq:Amperian-decomp-on-R} and~\eqref{eq:governing-for-smallest-evalue},
\begin{equation*}
A_{\lambda_1}(\cR) = A_{\lambda_1}^{\text{s}}(\cR).
\end{equation*}
It then follows from Proposition~\ref{prop:amp-sym-is-G-S} and the simplicity of the first Dirichlet eigenvalue of the Grad-Shafranov operator on the cross-section $\Sigma$ of $\cR$ that $\lambda_1^2 = \lambda^{\text{G-S}}_1$ and that the vector space $A_{\lambda_1}^{\text{s}}(\cR)$ is simple. Conclusion \ref{item:sym-conclusion} thus follows.

If $\lambda_1^{\text{s}} > \lambda_1^*$, then from Equations~\eqref{eq:Amperian-decomp-on-R} and~\eqref{eq:governing-for-smallest-evalue},
\begin{gather*}
A_{\lambda_1}(\cR) = E_{\lambda_1}(\cR) \cap \cB_K(\cR),
\end{gather*}
while $A_{\lambda_1}^\text{s}(\cR)$ is a trivial vector space. Conclusion \ref{item:antisym-conclusion} thus follows.
\end{proof}

\section{Proof of Theorem \ref{thm:sym-wins}}\label{sec:proof-of-1st-thm}

Proposition~\ref{prop:amp-sym-is-G-S} and Lemma~\ref{lem:basic-G-S-ineq} imply that
\begin{equation}\label{eq:thm1.2-sym-ineq}
\begin{split}
\lambda^{\text{s}}_1(\cR_{a,R}) = \sqrt{\lambda^{\text{G-S}}_1(D_{a,R})} &\leq \sqrt{\lambda_1^D(D_{a,R}) + \frac{3}{4(R+a)^2}}\\
&= \sqrt{\frac{j_{0,1}^2}{a^2} + \frac{3}{4(R+a)^2}},
\end{split}
\end{equation}
where $j_{0,1}$ denotes the first positive zero of the Bessel function $J_0$.

Now we focus on $\lambda^*_1(\cR_{a,R})$, the first positive antisymmetric curl eigenvalue on $\cR_{a,R}$. Recall that we denote by $K=\partial_\varphi$ the Killing field on $\cR_{a,R}$ as discussed in Section~\ref{sec:prelim}.
The following lemma, which is proved in Section~\ref{sec:antisymm-ineq-warped-prod}, establishes an equality enabling passage to an averaged metric on $\mathcal R_{a,R}$.

\begin{lemma}\label{lem:warped-prod-lower-bound}
Setting $r_{1/2} = \sqrt{R-a}\sqrt{R+a}$, consider the averaged metric
\begin{equation*}
g_{r_{1/2}} = dr^2 + dz^2 + r^2_{1/2}d\varphi^2    
\end{equation*}
on $\mathcal R_{a,R}$ (obviously, $K$ is a Killing field of this metric). Then, letting $\widetilde\lambda^{*}_1(\mathcal R_{a,R})$ be the first positive $K$-antisymmetric curl eigenvalue on $(\mathcal R_{a,R},g_{r_{1/2}})$, one has
\begin{equation*}
\lambda^{*}_1(\mathcal R_{a,R}) \geq \left(\frac{R-a}{R+a}\right)^{\frac{1}{2}}\widetilde\lambda^{*}_1(\mathcal R_{a,R}).
\end{equation*}
\end{lemma}

The following lemma establishes a lower bound for the first antisymmetric curl eigenvalue $\widetilde\lambda^{*}_1(\mathcal R_{a,R})$ (with the metric $g_{r_{1/2}}$); its proof is presented in Section \ref{sec:analysis-disk-anti-sym}.

\begin{lemma}\label{lem:applying-universal-disk-ineq}
There exists $j^* > j_{0,1}$ such that, for any $a > 0$ ,
\begin{equation*}
\widetilde \lambda_1^*(\mathcal R_{a,R}) \geq \frac{j^*}{a}.    
\end{equation*}
\end{lemma}

The previous lemmas imply that we can write
\begin{equation*}
\lambda^*_1(\cR_{a,R}) \geq \left(\frac{R-a}{R+a}\right)^{\frac{1}{2}}\frac{j^*}{a}.
\end{equation*}
Recalling \eqref{eq:thm1.2-sym-ineq}, we have altogether that, for all $0 < a < R$,
\begin{align*}
a \lambda^{\text{s}}_1(\cR_{a,R}) &\leq \sqrt{j_{0,1}^2 + \frac{3a^2}{4(R+a)^2}},\\
a \lambda^*_1(\cR_{a,R}) &\geq \left(\frac{R-a}{R+a}\right)^{\frac{1}{2}}j^*,
\end{align*}
but
\begin{equation*}
\lim_{a \to 0^+} \sqrt{j_{0,1}^2 + \frac{3a^2}{4(R+a)^2}}= j_{0,1} < j^* = \lim_{a \to 0^+} \left(\frac{R-a}{R+a}\right)^{\frac{1}{2}}j^*,
\end{equation*}
and so, for sufficiently small $a > 0$, one has that
\begin{equation*}
\lambda^{\text{s}}_1(\cR_{a,R}) < \lambda^*_1(\cR_{a,R}),    
\end{equation*}
and so by Lemma \ref{lem:governing-principle} we are done.

\section{Proof of Theorem \ref{thm:antisym-wins}}\label{sec:proof-of-2nd-thm}

Let $L > 0$ and $0 < a < b$. Let $n\in \NN$ and $\cR \subset \RR^3$ be a smooth solid torus of revolution satisfying the set-theoretic relations in the statement of the theorem. Setting the rectangle
\begin{equation*}
R \coloneqq (a,b) \times (-L_{n+1}/2,L_{n+1}/2),    
\end{equation*}
with $L_n\coloneqq nL$, one has that $\Int \Sigma \subset R$, and hence Proposition~\ref{prop:amp-sym-is-G-S} and Lemma~\ref{lem:basic-G-S-ineq} imply
\begin{equation}\label{eq:thm1.3-sym-ineq}
\begin{split}
\lambda^{\text{s}}_1(\cR) = \sqrt{\lambda^{\text{G-S}}_1(\Sigma)} \geq \sqrt{\lambda^{\text{G-S}}_1(R)} &\geq \sqrt{\lambda^{D}_1(R) + \frac{3}{4b^2}}\\
&= \sqrt{\frac{\pi^2}{(b-a)^2} + \frac{3}{4b^2} + \frac{\pi^2}{L_{n+1}^2}}.
\end{split}
\end{equation}

We now focus on the first (positive) antisymmetric curl eigenvalue $\lambda^*_1(\cR)$.
It is convenient to introduce the Riemannian 3-manifold with boundary
\begin{equation*}
M_{a,b,L_n} \coloneqq A_{a,b} \times \RR/L_n\ZZ,
\end{equation*}
with the metric
\begin{equation*}
g = dx^2 + dy^2 + ds^2,
\end{equation*}
and the Killing field
\begin{equation*}
P = -y\partial_x + x\partial_y,    
\end{equation*}
where $(x,y)\in A_{a,b}$ and $s\in \mathbb R/L_n\ZZ$ denote Cartesian coordinates. Then, we have the following lemma proven in Section \ref{sec:upper-bound-antisym-rectangular-section}.

\begin{lemma}\label{lem:antisym-inequality}
Let $\lambda^P_1(M_{a,b,L_n})$ denote the first $P$-antisymmetric curl eigenvalue on $M_{a,b,L_n}$. Then,
\begin{equation*}
\lambda^*_1(\cR) \leq \left(\frac{\pi^2 b^2}{L_n^2} + 1\right)\lambda^P_1(M_{a,b,L_n}).     
\end{equation*}
\end{lemma}

Lastly, we have the following lemma about the existence of a certain $P$-antisymmetric curl eigenvalue. Its proof is found in Section~\ref{sec:analysis-annulus-anti-sym}. 

\begin{lemma}\label{lem:existence-of-small-eigenvalue}
There exist $0 < a < b$ and $L > 0$ such that for all $n \in \NN$, the number
\begin{equation*}
\lambda = \frac{\pi}{b-a}    
\end{equation*}
is a $P$-antisymmetric curl eigenvalue on $M_{a,b,L_n}$.
\end{lemma}

We are now ready to complete the proof of Theorem~\ref{thm:antisym-wins}. Fixing $b>0$, take numbers $0 < a < b$ and $L > 0$ satisfying Lemma~\ref{lem:existence-of-small-eigenvalue}. Let $n\in \NN$ and $\cR \subset \RR^3$ be a smooth solid torus of revolution satisfying the set-theoretic relations in the theorem. Then, from Equation~\eqref{eq:thm1.3-sym-ineq}, we have 
\begin{equation}\label{eq:specific-sym-upperbound}
\lambda^{\text{s}}_1(\cR) > \sqrt{\frac{\pi^2}{(b-a)^2} + \frac{3}{4b^2}}.
\end{equation}
On the other hand, $\pi/(b-a)$ is a $P$-antisymmetric curl eigenvalue on $M_{a,b,L_n}$ and hence, by Lemma \ref{lem:antisym-inequality}, 
\begin{equation}\label{eq:specific-antisym-upperbound}
\lambda^*_1(\cR) \leq \left(\frac{\pi^2 b^2}{L^2_{n}} + 1\right) \frac{\pi}{b-a}.
\end{equation}
Fixing some $N \in \NN$ such that
\begin{equation*}
\left(\frac{\pi^2 b^2}{L^2_{N}} + 1\right)\frac{\pi}{b-a} \leq \sqrt{\frac{\pi^2}{(b-a)^2} + \frac{3}{4b^2}}, 
\end{equation*}
we infer from \eqref{eq:specific-sym-upperbound} and \eqref{eq:specific-antisym-upperbound},
\begin{equation*}
\lambda^{\text{s}}_1(\cR) > \left(\frac{\pi^2 b^2}{L^2_{n}} + 1\right) \frac{\pi}{b-a} \geq \lambda^*_1(\cR).
\end{equation*}
for all $n\geq N$, and we are done by Lemma~\ref{lem:governing-principle}.

\section{Proof of Lemma \ref{lem:warped-prod-lower-bound}}\label{sec:antisymm-ineq-warped-prod}

In this section we prove the following estimate:
\begin{equation*}
C^{-1}\widetilde \lambda_1^*(\mathcal R_{a,R})\leq  \lambda_1^*(\mathcal R_{a,R})\leq C\widetilde \lambda_1^*(\mathcal R_{a,R}),
\end{equation*}
with
\begin{equation*}
C\coloneqq\Big(\frac{R+a}{R-a}\Big)^{1/2}.
\end{equation*}

As before, $\mathcal B_K(\mathcal R_{a,R})$ is the space of $K$-antisymmetric solenoidal vector fields on $\cR_{a,R}$ (tangent to the boundary), and we use the notation $\mathcal B_{K,1/2}(\mathcal R_{a,R})$ for the corresponding space when $\mathcal R_{a,R}$ is endowed with the averaged metric $g_{r_{1/2}}$. It is straightforward to check that if $B\in \mathcal B_K(\mathcal R_{a,R})$ then $\tau(B)\coloneqq fB\in \mathcal B_{K,1/2}(\mathcal R_{a,R})$ for the function 
\begin{equation*}
f\coloneqq\frac{r}{r_{1/2}}.
\end{equation*}
Moreover, it is clear that
\begin{equation*}
i_B\,rdr \wedge d\varphi \wedge dz = i_{\tau(B)}\,r_{1/2}dr \wedge d\varphi \wedge dz,
\end{equation*}
so denoting the helicity of $\tau(B)$ with respect to the metric $g_{r_{1/2}}$ as $\mathcal H_{1/2}(\tau(B))$ we easily deduce that
\begin{equation}\label{eq:equal-helicity}
\cH_{1/2}(\tau(B)) =\cH(B).
\end{equation}
Now we compute
\begin{equation*}
(B,B)_{L^2} = \int_{\mathcal R_{a,R}} |B|^2 rdrd\varphi dz = \int_{\mathcal R_{a,R}} f |B|^2 r_{1/2}dr d\varphi dz.
\end{equation*}
Moreover, writing $B$ in coordinates as $B=B_r\partial_r+B_\varphi \partial_\varphi+B_z\partial_z$, we have
\begin{equation*}
f |B|^2 = \frac{r_{1/2}}{r}~\frac{r^2}{r_{1/2}^2} (B_r^2 + B_z^2) + \frac{r}{r_{1/2}}~r^2B_\varphi^2,
\end{equation*}
which we can compare with
\begin{equation*}
g_{r_{1/2}}(\tau(B),\tau(B)) = \frac{r^2}{r_{1/2}^2}\left(B_r^2 +B_z^2\right)+r^2B_\varphi^2.
\end{equation*}
Observing that
\begin{align*}
\max\left(\frac{r_{1/2}}{r}\right) = C = \max\left(\frac{r}{r_{1/2}}\right),\\
\min\left(\frac{r_{1/2}}{r}\right) = C^{-1} =  \min\left(\frac{r}{r_{1/2}}\right),
\end{align*}
we obtain the pointwise inequality
\begin{equation*}
C^{-1} g_{r_{1/2}}(\tau(B),\tau(B)) \leq f |B|^2 \leq C g_{r_{1/2}}(\tau(B),\tau(B)), 
\end{equation*}
and so, multiplying by $r_{1/2}$ and integrating on $\mathcal R_{a,R}$,
\begin{equation*}
C^{-1} \|\tau(B)\|^2_{L^2,r_{1/2}} \leq \|B\|^2_{L^2} \leq C \|\tau(B)\|^2_{L^2,r_{1/2}}, 
\end{equation*}
where the subindex $r_{1/2}$ in the scalar product denotes that it is computed with respect to the metric $g_{r_{1/2}}$. Taking Equation \eqref{eq:equal-helicity} into account, we therefore obtain
\begin{equation}\label{eq:Rayleigh-quotient-estimates}
C^{-1} \frac{\cH_{{1/2}}(\tau(B))}{\|\tau(B)\|^2_{L^2,r_{1/2}}} \leq \frac{\cH(B)}{\|B\|^2_{L^2}} \leq C \frac{\cH_{{1/2}}(\tau(B))}{\|\tau(B)\|^2_{L^2,r_{1/2}}}.  
\end{equation}
Finally, taking maxima over $B\in \mathcal B_K(\mathcal R_{a,R}) \setminus \{0\}$, and using the fact that $\tau$ defines a linear isomorphism between $B_K(\mathcal R_{a,R})$ and $\mathcal B_{K,1/2}(\mathcal R_{a,R})$, the estimate~\eqref{eq:Rayleigh-quotient-estimates} and the variational characterization of the first antisymmetric curl eigenvalue  (cf. Corollary~\ref{cor:antisym-spectral}) yield
\begin{equation*}
C^{-1} \left(\widetilde\lambda_1^*(\mathcal R_{a,R})\right)^{-1} \leq \frac{1}{\lambda_1^*(\mathcal R_{a,R})} \leq C \left(\widetilde\lambda_1^*(\mathcal R_{a,R})\right)^{-1},
\end{equation*}
which proves the result.

\section{A complex-valued formulation for antisymmetric curl eigenfields}\label{sec:complex-antisym-fields}

The purpose of this section is to introduce some complex algebra associated with the antisymmetric curl spectrum used in the proof of the remaining lemmas. Let $(M,g)$ be a compact  Riemannian 3-manifold with boundary, and let $\cX(M,\CC)$ denote the space of smooth complex-valued vector fields on $M$. For any $Z \in \cX(M,\CC)$, there exist unique $X,Y \in \cX(M)$ such that $Z = X + iY$ and we set
\begin{equation*}
\Re Z = X, \qquad \Im Z = Y, \qquad \overline{Z} = X-iY.
\end{equation*}
Let $g_\CC$ denote inner product such that, for vector fields $Z,W \in \cX(M,\CC)$,
\begin{equation*}
g_\CC(Z,W) \coloneqq g(Z,\overline{W}).    
\end{equation*}
Letting $\mu$ denote the Riemannian volume form on $M$, for $Z,W \in \cX(M,\CC)$ we set
\begin{equation*}
(Z,W)_{L^2} = \int_M g_\CC(Z,W)\mu = \int_M g(Z,\overline{W})\mu.\end{equation*}
For an $\RR$-linear subspace $S \subset \cX(M)$, we let $S_\CC$ denote the $\CC$-linear subspace of $\cX(M,\CC)$ given by 
\begin{equation*}
S_\CC \coloneqq \{Z \in \cX(M,\CC) : \Re Z \in S \text{ and } \Im Z \in S\}.    
\end{equation*}
Let
\begin{equation*}
\cB_{\CC,\text{FF}}(M) = (\cB_{\text{FF}}(M))_{\CC},
\end{equation*}
where $\cB_{\text{FF}}(M)$ is defined in Appendix~\ref{app:func-analysis-antisym-spec}. It is easy to prove but nevertheless important to observe that helicity extends nicely to the complex setting.

\begin{lemma}\label{lem:helicity-is-nice}
The $\CC$-linear extension $\curl^{-1} : \cB_{\CC,\text{FF}}(M) \to \cB_{\CC,\text{FF}}(M)$ of $\curl^{-1}$ on $\cB_{\text{FF}}(M)$ is symmetric with respect to the (complex) $L^2$ inner product. In particular, the function $\cH_\CC$ on $\cB_{\CC,\text{FF}}(M)$ given by
\begin{equation*}
\cH_{\CC}(B) = (B,\curl^{-1}B)_{L^2}   
\end{equation*}
is real-valued. For any $B \in \cB_{\CC,\text{FF}}(M)$ and any $A \in \cX(M,\CC)$ with $\curl A = B$, one has that
\begin{equation*}
\cH_\CC(B) = (B,A)_{L^2}.    
\end{equation*}
Moreover, for all $B \in \cB_{\CC,\text{FF}}(M)$, if $\cH : \cB_{\text{FF}}(M) \to \RR$ denotes helicity on $\cB_{\text{FF}}(M)$,
\begin{equation*}
\cH_\CC(B) = \cH(\Re B) + \cH(\Im B).    
\end{equation*}
\end{lemma}

Henceforth assume that $(M,g)$ admits a non-trivial Killing field $K \in \cX(M)$ tangent to the boundary. For $\lambda \in \RR \setminus \{0\}$ and $k \in \RR$, we set
\begin{align*}
E_{\CC,\lambda}(M) &= \{B \in \cX(M,\CC) : \curl B = \lambda B,\, B \cdot n = 0\}\\
E_{\CC,\lambda,K,k}(M) &= \{B \in E_{\CC,\lambda}(M) : \cL_K B = ik B\}.
\end{align*}

Then, it is easy to prove the following. In the statement, we use the notation introduced in Proposition~\ref{prop:decomp}.

\begin{lemma}\label{lem:complex-decomp}
For $\lambda \in \RR \setminus \{0\}$ and $k > 0$, one has the $L^2$-orthogonal decomposition
\begin{equation*}
(E_{\lambda,K,k}(M))_{\CC} = E_{\CC,\lambda,K,k}(M) \oplus E_{\CC,\lambda,K,-k}(M). 
\end{equation*}
\end{lemma}

We now state a useful vector potential formula for vector fields which generalise those in Lemma \ref{lem:complex-decomp}.

\begin{lemma}\label{lem:vector-pot-and-antisymm-lamda-ineq}
Let $k \in \RR\setminus \{0\}$ and $B \in \cX(M,\CC)$ be a complex vector field satisfying
\begin{equation*}
\Div B = 0, \qquad \cL_K B = ik B, \qquad B \cdot n = 0.    
\end{equation*}
Then $B \in \cB_{\CC,\text{FF}}(M)$. Letting $\lambda^K_1(M)$ and $\lambda^K_{-1}(M)$ respectively denote the first positive and first negative $K$-antisymmetric curl eigenvalues, we have that
\begin{equation}\label{eq:helicity-inequality}
\frac{1}{\lambda^K_{-1}(M)} \leq \frac{\cH_\CC(B)}{\|B\|_{L^2}^2} \leq \frac{1}{\lambda^K_1(M)}.
\end{equation}
In addition, the vector field
\begin{equation*}
A = -\frac{i}{k} B \times K,    
\end{equation*}
is a normal-pointing vector potential for $B$. In particular, the helicity of $B$ is given by
\begin{equation}\label{eq:helicity-formula}
\cH_\CC(B) = \frac{i}{k}(B,B\times K)_{L^2}.    
\end{equation}
If $\lambda \neq 0$ and $B$ satisfies that $\curl B = \lambda B$, then
\begin{equation}\label{eq:Beltrami-helicity}
\frac{\cH_\CC(B)}{\|B\|_{L^2}^2} = \frac{1}{\lambda}.    
\end{equation}
\end{lemma}

\begin{proof}
Writing $B = X + iY$, we have that $X,Y \in \cB(M)$. Moreover, from $\cL_K B = ik B$ we have that $\cL_K^2 B = -k^2 B$ so that in fact $X,Y \in \cB_K(M)$. Hence, $B \in \cB_{\CC,\text{FF}}(M)$ by Lemma~\ref{lem:compatibility-of-antisym}. Moreover, because $X,Y \in \cB_{K}(M)$, the estimate~\eqref{eq:helicity-inequality} follows from Corollary \ref{cor:antisym-spectral} and the equalities
\begin{equation*}
\cH_\CC(B) = \cH(X) + \cH(Y), \qquad \|B\|_{L^2}^2 = \|X\|_{L^2}^2 + \|Y\|_{L^2}^2.
\end{equation*}
The first equality coming from Lemma \ref{lem:helicity-is-nice}. Concerning the vector field $A$,
\begin{align*}
\curl A = -\frac{i}{k} \curl(B \times K)
= -\frac{i}{k}\cL_K B=B.
\end{align*}
Hence, Equation \eqref{eq:helicity-formula} holds. If $\lambda \neq 0$ and $B$ satisfies that $\curl B = \lambda B$, then Equation \eqref{eq:Beltrami-helicity} follows from Lemma \ref{lem:helicity-is-nice} and the fact that $A' = B/\lambda$ is a vector potential for $B$.
\end{proof}

Lastly we introduce some useful notation for flat tubes. For $0 < a < b$ we set
\begin{align*}
D_a &= \{(x,y) \in \RR^2 : x^2+y^2 \leq a^2\},\\  
A_{a,b} &= \{(x,y) \in \RR^2 : a^2 \leq x^2+y^2 \leq b^2\}. 
\end{align*}
For $L > 0$ and $\Sigma \in \{D_a,A_{a,b}\}$, we set
\begin{equation*}
M_{\Sigma,L} = \Sigma \times \RR/L\ZZ.     
\end{equation*}
Let $(x,y) \in \Sigma$ and $s \in \RR/L\ZZ$ denote the usual Cartesian coordinates and consider the flat metric $g$ on $M_{\Sigma,L}$ and associated Killing fields $T$ and $P$ given by
\begin{equation*}
g = dx^2 + dy^2 + ds^2, \quad T = \partial_s, \quad P = -y\partial_x + x\partial_y.
\end{equation*}
Note that the set of eigenvalues of $\cL_T$ and $\cL_P$ acting on either $C^{\infty}(M,\CC)$ or $\cX(M,\CC)$ are respectively given by $(2\pi/L)\ZZ$ and $\ZZ$. With this in mind, Lemma \ref{lem:complex-decomp} and Proposition \ref{prop:decomp} imply the following decomposition.

\begin{lemma}\label{lem:new-decomp}
Let $\lambda \in \RR \setminus \{0\}$. Then, the eigenspace $E_{\CC,\lambda}(M_{\Sigma,L})$ has the $L^2$-orthogonal decomposition
\begin{equation*}
E_{\CC,\lambda}(M_{\Sigma,L}) = \bigoplus_{i=1}^N E_{\CC,\lambda,\ell_i,m_i}(M_{\Sigma,L}),
\end{equation*}
for some $N \in \NN$ and numbers $\ell_1,...,\ell_N \in (2\pi/L)\ZZ$ and $m_1,...,m_N \in \ZZ$, where for $\ell \in (2\pi/L)\ZZ$ and $m \in \ZZ$, we have set
\begin{equation*}
E_{\CC,\lambda,\ell,m}(M_{\Sigma,L})\coloneqq  \{B \in E_{\CC,\lambda}(M) : \cL_T B = i\ell B,\, \cL_P B = im B\}.
\end{equation*}
\end{lemma}
For the remaining sections, we use the notation $M_{a,L} \coloneqq M_{D_a,L}$ and $M_{a,b,L} \coloneqq M_{A_{a,b},L}$.

\section{Proof of Lemma \ref{lem:applying-universal-disk-ineq}}\label{sec:analysis-disk-anti-sym}

For $a > 0, L > 0$, we first consider the Riemannian 3-manifold $M_{a,L}$ with boundary as in Section \ref{sec:complex-antisym-fields}. The proof of Lemma \ref{lem:applying-universal-disk-ineq} is done through analysing the appropriate eigenvalue equation. Specifically, from Lemma \ref{lem:eigen-equations} in Appendix \ref{app:flat-tubes}, given $\ell \in 2\pi/L \ZZ$ and $m \in \ZZ$, there exist a nontrivial $B$ and $\lambda \neq 0$ with
\begin{align*}
B \in E_{\CC,\lambda,\ell,m}(M_{a,L}),
\end{align*}
if and only if $\lambda^2 > \ell^2$ and, setting $\mu \coloneqq  \sqrt{\lambda^2-\ell^2}$, $\lambda$ satisfies
\begin{equation}\label{eq:eigenvalue-equation-disk-again}
\frac{\lambda m}{a} J_m(\mu a) + \ell \mu J_m'(\mu a) = 0.    
\end{equation}

We introduce an algebraic manipulation which will prove useful here as well as in Section \ref{sec:analysis-annulus-anti-sym}. First, recall that for $m \in \RR$, if $Z_m \in \{J_m,Y_m\}$ denotes a Bessel function of order $m$ of the first or second kind, then for $x \neq 0$,
\begin{equation}\label{eq:Bessel-recurrence-rel}
Z_m'(x) = \frac{Z_{m-1}(x) - Z_{m+1}(x)}{2}, \qquad \frac{m Z_m(x)}{x} = \frac{Z_{m-1}(x) + Z_{m+1}(x)}{2}.
\end{equation}
In particular, Equation \eqref{eq:Bessel-recurrence-rel} implies that, for any $\ell \in \RR$, $\lambda \in \RR$, $c\neq 0$ and $\mu \neq 0$,
\begin{equation}\label{eq:bessel-manipulation}
\frac{\lambda m}{c} Z_m(\mu c) + \ell \mu Z_m'(\mu c)
= \frac{\mu}{2}((\lambda+\ell) Z_{m-1}(\mu c) + (\lambda-\ell) Z_{m+1}(\mu c)).
  \end{equation}

Setting $U \coloneqq  \{(\alpha,\kappa) \in \RR^2 : \kappa^2 > \alpha^2\}$, for $m \in \ZZ$, we consider the functions $F_m : U \to \RR$ given by
\begin{equation*}
F_m(\alpha,\kappa) \coloneqq  \frac{\kappa+\alpha}{\kappa-\alpha}J_{m-1}(\sqrt{\kappa^2-\alpha^2}) + J_{m+1}(\sqrt{\kappa^2-\alpha^2}).
\end{equation*}
Our primary interest is analysing the zeros of these functions. This is because, through Equation \eqref{eq:bessel-manipulation}, we see that, for $\lambda \in \RR\setminus\{0\}$, $\ell \in 2\pi/L \ZZ$, $m \in \ZZ$, one has that $\lambda^2 > \ell^2$ and Equation \eqref{eq:eigenvalue-equation-disk-again} is satisfied if and only if $(a\ell,a\lambda) \in U$ and
\begin{equation}\label{eq:connection-between-kappa-and-lambda}
F_m(a \ell,a \lambda) = 0.   
\end{equation}

The following three elementary lemmas allow us to study the zeros of the functions $F_m$. In the statements, for $m \in \NN_0 = \{0,1,2,...\}$ and $k \in \NN$, we let $j_{m,k}$ denote the $k^{\text{th}}$ positive zero of the Bessel function $J_m$.

\begin{lemma}\label{lem:m-is-zero}
For $\alpha \neq 0$, we have $(\alpha,\kappa) \in F_0^{-1}(0)$ if and only if there exists $k \in \NN$ with
\begin{equation*}
\kappa = \pm\sqrt{j_{1,k}^2 + \alpha^2}.
\end{equation*}
\end{lemma}

\begin{proof}
The result follows noticing that
\begin{equation*}
F_0(\alpha,\kappa) = \left(1-\frac{\kappa+\alpha}{\kappa-\alpha}\right)J_{1}(\sqrt{\kappa^2-\alpha^2}),
\end{equation*}
and that $\alpha \neq 0$.
\end{proof}

\begin{lemma}\label{lem:main-bounding-work}
For $m \in \NN$, setting
\begin{equation*}
j^*_m \coloneqq \inf\{\kappa > 0 : (\alpha,\kappa) \in F_m^{-1}(0) \text{ for some } \alpha \in \RR \},   
\end{equation*}
one has that
\begin{equation}\label{eq:inf-upper-bound}
j^*_m > j_{m-1,1}.
\end{equation}
\end{lemma}

\begin{proof}
For $\kappa > 0$ if $(\alpha,\kappa) \in U$ and
\begin{equation*}
\kappa \leq \sqrt{j_{m-1,1}^2+\alpha^2},
\end{equation*}
then $0 < \sqrt{\kappa^2-\alpha^2} \leq j_{m-1,1}$ so that, because $j_{m-1,1} < j_{m+1,1}$,
\begin{equation*}
J_{m-1}(\sqrt{\kappa^2-\alpha^2}) \geq 0, \qquad J_{m+1}(\sqrt{\kappa^2-\alpha^2}) > 0,
\end{equation*}
and hence $F_m(\alpha,\kappa) > 0$. Therefore, for any $(\alpha,\kappa) \in F_m^{-1}(0)$,
\begin{equation}\label{eq:jm-11-bound-for-kappa}
\kappa > \sqrt{j_{m-1,1}^2 + \alpha^2}.
\end{equation}

Next, we observe that the function $F_m(\alpha,\cdot)$ is analytic in the closed set $S_\alpha:=\{\kappa \in \RR : \kappa \geq \sqrt{j_{m-1,1}^2 + \alpha^2}\}$, so its zero set consists of (countably many) isolated points. When $\alpha=0$ and $\kappa>0$, with the help of Equation \eqref{eq:Bessel-recurrence-rel} we see that
\begin{equation}\label{eq:alpha-zero-case}
F_m(0,\kappa) =  \frac{2m}{\kappa}J_m(\kappa). 
\end{equation}
We then conclude that $(0,\kappa)\in F_m^{-1}(0)$ if and only if $\kappa$ is a zero of $J_m$, and hence the first zero of $F_m(0,\kappa)$ in $S_0$ is $j_{m,1}>j_{m-1,1}$, which is nondegenerate. Finally, the continuity of nondegenerate zeros implies that the first positive zero of $F_m(\alpha,\cdot)$ is close to $j_{m,1}$ provided that $|\alpha|$ is small enough. While for $|\alpha|\geq \alpha_0>0$, the positive zeros of $F_m(\alpha,\cdot)$ are bounded below by $\sqrt{j_{m-1,1}^2+\alpha_0^2}$ from Equation~\eqref{eq:jm-11-bound-for-kappa}. The result follows.
\end{proof}

\begin{lemma}\label{lem:universal-bound-for-kappa}
Set $j^* = \min\{j^*_1,j_{1,1}\}$ with $j^*_1$ is as per Lemma \ref{lem:main-bounding-work}. Then, $j^* > j_{0,1}$ and for $m \in \ZZ$ and $(\alpha,\kappa) \in F^{-1}_m(0)$ with $\alpha \neq 0$,
\begin{equation}\label{eq:main-goal-for-kappa}
|\kappa| \geq j^*.
\end{equation}
\end{lemma}

\begin{proof}
We know that $j^*_1 > j_{0,1}$ and $j_{1,1} > j_{0,1}$, so $j^* > j_{0,1}$. To prove \eqref{eq:main-goal-for-kappa}, first note that for $m \in \ZZ$ and $(\alpha,\kappa) \in U$, we have that $(-\alpha,-\kappa),(-\alpha,\kappa) \in U$ and 
\begin{equation*}
\begin{split}
F_m(-\alpha,-\kappa) &= F_m(\alpha,\kappa),\\
F_{-m}(\alpha,\kappa) &= (-1)^{m+1} \frac{\kappa+\alpha}{\kappa-\alpha}F_m(-\alpha,\kappa).
\end{split}
\end{equation*}

Hence, it suffices to prove Equation \eqref{eq:main-goal-for-kappa} in the case that $m \in \NN_0$ and $\kappa > 0$. If $m = 0$, by Lemma~\ref{lem:m-is-zero} we know that the first positive zero of $F_0(\alpha,\cdot)$ is $\sqrt{j_{1,1}^2 + \alpha^2}$, which is obviously bigger than $j^*$ when $\alpha\neq 0$.

If $m =1$, by Lemma~\ref{lem:main-bounding-work} we see that
$\kappa\geq j^*_1 \geq j^*$, and if $m \geq 2$, we have
\begin{equation*}
\kappa\geq j^*_{m} > j_{m-1,1} \geq j_{1,1} \geq j^*,    
\end{equation*}
which completes the proof of~\eqref{eq:main-goal-for-kappa}.
\end{proof}

We are now ready to complete the proof of Lemma \ref{lem:applying-universal-disk-ineq} (recall the notation from Section \ref{sec:proof-of-1st-thm}). Setting $L = 2\pi r_{1/2}$, it is straightforward to check that $(x,y,s) \mapsto (r=x+R,z=y,\varphi=r_{1/2}^{-1}s)$ is an isometry between $M_{a,L}$ and $(\cR_{a,R},g_{r_{1/2}})$, and the Killing $T$ becomes the Killing $K$ in $(\cR_{a,R},g_{r_{1/2}})$, up to an inessential constant factor. Accordingly, if $\lambda > 0$ denotes the first positive $T$-antisymmetric curl eigenvalue on $M_{a,L}$, then
\begin{equation*}
\widetilde \lambda_1^*(\mathcal R_{a,R}) = \lambda.
\end{equation*}
From Proposition \ref{prop:decomp} and lemmas \ref{lem:complex-decomp} and \ref{lem:new-decomp}, we infer that there exist $0 \neq \ell \in 2\pi/L \ZZ$, $m \in \ZZ$, such that $E_{\CC,\lambda,\ell,m}(M_{a,L}) \neq \{0\}$. As discussed in the beginning of this section, this implies $\lambda^2 > \ell^2$ and that $\lambda$ solves Equation \eqref{eq:eigenvalue-equation-disk-again}, which in turn means that $(a\ell,a\lambda) \in U$ solves Equation \eqref{eq:connection-between-kappa-and-lambda} and so, by Lemma \ref{lem:universal-bound-for-kappa},
\begin{equation*}
\widetilde \lambda_1^*(\mathcal R_{a,R}) = \lambda \geq \frac{j^*}{a}.
\end{equation*}

\section{Proof of Lemma \ref{lem:antisym-inequality}}\label{sec:upper-bound-antisym-rectangular-section}

Before considering domains in Euclidean space, we have the following lemma for constructing appropriate test functions. Let $0 < a < b$ and $L > 0$ and consider the Riemannian 3-manifold $M_{a,b,L}$ with boundary as in Section \ref{sec:complex-antisym-fields} and recall the notation from that section.

\begin{lemma}\label{lem:calculations-for-modded-field}
Let $f\in C^\infty(\mathbb R/L\mathbb Z)$ and $0 \neq m \in \ZZ$, $\ell \in (2\pi/L) \ZZ$. Let $B \in \cX(M_{a,b,L},\CC)$ be such that
\begin{equation*}
\Div B = 0,\quad \cL_T B = i\ell B,\quad \cL_P B = imB, \quad B \cdot n = 0.
\end{equation*}
Define
\begin{align*}
A &= -\frac{i}{m} B \times P,\\
A_f &= f(s) A,\\
B_f &= \curl A_f.
\end{align*}
Then,
\begin{equation}\label{eq:Bf-is-admissible}
\Div B_f = 0, \quad \cL_P B_f = i m B_f, \quad B_f \cdot n = 0,
\end{equation}
and
\begin{align}
\cH_\CC(B_f) &= \frac{1}{L}\int_{\RR/L\ZZ} f^2 ds~\cH_\CC(B),\label{eq:helicity-of-modded-B}\\
\|B_f\|_{L^2}^2 &\leq \left(\frac{b^2}{m^2L}\int_{\RR/L\ZZ}(f')^2ds + \frac{1}{L}\int_{\RR/L\ZZ} f^2 ds\right)\|B\|_{L^2}^2.\label{eq:energy-of-modded-B}
\end{align}
\end{lemma}

\begin{proof}
From Lemma \ref{lem:vector-pot-and-antisymm-lamda-ineq} we have that $\curl A = B$ and that $A$ is normal-pointing. Hence, $A_f = f(s) A$ is normal-pointing. Moreover, because $\cL_P B = imB$, we have that $\cL_P A = im A$. Moreover, $ds(P) = 0$. Therefore,
\begin{equation*}
\cL_P A_f = im A_f,
\end{equation*}
and hence $B_f = \curl A_f$ satisfies Equation \eqref{eq:Bf-is-admissible}. Now,
\begin{equation*}
B_f = \curl A_f = f'(s) T \times A + f(s) B.
\end{equation*}
Moreover,
\begin{align*}
T \times A = \frac{i}{m}(B \cdot T)P,
\end{align*}
and hence
\begin{equation}\label{eq:expression-for-Bf}
B_f = \frac{i}{m} f'(s)(B \cdot T)P + f(s)B.
\end{equation}
Now, because $f$ and $P$ are real, we have that
\begin{equation*}
\overline{A_f} =  \frac{i}{m} f(s) \overline{B} \times P,
\end{equation*}
and in particular, $P \cdot \overline{A_f} = 0$, so Equation \eqref{eq:expression-for-Bf} implies that
\begin{equation}\label{eq:modified-helicity-density}
B_f \cdot \overline{A_f} = f(s)B \cdot A_f = f^2(s) B \cdot \overline{A}.
\end{equation}

We will henceforth make repeated use of the following observation. If $G \in C^{\infty}(M,\CC)$ is a function with $\cL_T G = 0$, and $F \in C^{\infty}(\RR/L\ZZ,\CC)$, then $G \equiv G(x,y)$ and
\begin{equation}\label{eq:Fubini}
\begin{split}
\int_M F(s)G \mu = \frac{1}{L}\int_{\RR/L\ZZ}Fds \int_M G \mu.
\end{split}
\end{equation}

Now, observe that
\begin{equation}\label{eq:deducing-dependence}
\cL_T(B \cdot \overline{A}) = i\ell B \cdot \overline{A} + B \cdot (-i\ell \overline{A}) = 0.
\end{equation}
By Equation \eqref{eq:modified-helicity-density}, we may therefore integrate $B_f \cdot \overline{A_f}$ using Equation \eqref{eq:Fubini}. This yields Equation \eqref{eq:helicity-of-modded-B}.

Concerning the energy of $B_f$, we use Equation \eqref{eq:expression-for-Bf} and expand
\begin{align*}
|B_f|^2 
= \frac{1}{m^2} f'(s)^2 r^2|B \cdot T|^2 + 2\Re \left(\frac{i}{m}f'(s)f(s) (B \cdot T)(\overline{B} \cdot P)\right) + f(s)^2|B|^2,
\end{align*}
where we used that $f$ is real-valued and that $|P|^2 = x^2+y^2 = r^2$. It is useful to note that, with the same reasoning as that in Equation \eqref{eq:deducing-dependence}, we may deduce that
\begin{equation}\label{eq:remaining-z-independent-quantities}
\begin{aligned}
\cL_T (r^2|B \cdot T|^2)  = 0,\\
\cL_T ((B \cdot T)(\overline{B} \cdot P)) = 0,\\
\cL_T |B|^2 = 0.
\end{aligned}
\end{equation}
Hence, applying Equation \eqref{eq:Fubini}, we deduce
\begin{equation*}
\int_M f'(s)f(s) (B \cdot T)(\overline{B} \cdot P) \mu = \frac{1}{L}\int_{\RR/L\ZZ} f' f ds \int_M (B \cdot T)(\overline{B} \cdot P) \mu = 0.
\end{equation*}

The cross-term of $|B_f|^2$ therefore integrates to zero, and so we can write that
\begin{align*}
\|B_f\|_{L^2}^2 = \frac{1}{m^2}\int_M f'(s)^2r^2|B\cdot T|^2 \mu + \int_{M} f^2(s)|B|^2 \mu.     
\end{align*}
Once again, using equations \eqref{eq:remaining-z-independent-quantities} and \eqref{eq:Fubini},
\begin{equation}\label{eq:energy-formula}
\|B_f\|_{L^2}^2 = \frac{1}{m^2L}\int_{\RR/L\ZZ}(f')^2ds \int_Mr^2|B\cdot T|^2\mu + \frac{1}{L}\int_{\RR/L\ZZ} f^2 ds~\|B\|_{L^2}^2. 
\end{equation}
Lastly, because $|T| = |\partial_s| = 1$ and $a \leq r \leq b$, we can say that $|B \cdot T|^2 \leq |B|^2$ and $r^2 \leq b^2$. Applying these inequalities to Equation \eqref{eq:energy-formula} yields \eqref{eq:energy-of-modded-B}.
\end{proof}

We continue with the proof of Lemma \ref{lem:antisym-inequality}. Let $\cR \subset \RR^3$ is a smooth solid torus of revolution satisfying
\begin{equation*}
\cC_{a,b,L} \subset \cR, \qquad \partial A_{a,b} \times [-L/2,L/2] \subset \partial \cR.
\end{equation*}
Recall the Killing field
\begin{equation*}
K = -Y\partial_X + X\partial_Y
\end{equation*}
on $\cR$ tangent to $\partial \cR$ where $(X,Y,Z) \in \cR$ denote Cartesian coordinates.

Let $\pi : A_{a,b} \times \RR \to M_{a,b,L}$ denote the projection given by $\pi(x,y,z) = (x,y,[z])$,  which restricts to an isometric embedding of
\begin{equation*}
U_{a,b,L} \coloneqq  A_{a,b} \times (-L/2,L/2) \subset \RR^3.
\end{equation*}
It is clear that $U_{a,b,L}\subset \cR$ and $\partial U_{a,b,L} \subset \partial \cR$.

Let $\hat{f}_\infty\coloneqq \sqrt{2}\cos(\pi t/L)$ be the first Dirichlet Laplacian eigenfunction on $[-L/2,L/2]$.
Then, we observe that
\begin{equation*}
\frac{1}{L}\int_{-L/2}^{L/2} \hat{f}_{\infty}(t)^2 dt = 1, \qquad \frac{1}{L}\int_{-L/2}^{L/2} \hat{f}'_{\infty}(t)^2 dt = \frac{\pi^2}{L^2}. 
\end{equation*}
Obviously, there exists a sequence $\{\hat{f}_j : [-L/2,L/2] \to \RR\}_{j \in \NN}$ of compactly supported smooth functions on $(-L/2,L/2)$ such that $\hat{f}_j \to \hat{f}$ in the $H^1$ norm on $[-L/2,L/2]$. 

Let $f_j \in C^{\infty}(\RR/L\ZZ)$ be the unique functions such that $f_j([t]) = \hat{f}_j(t)$ for $t \in [-L/2,L/2]$. The $H^1$-convergence
directly implies the limits
\begin{equation}\label{eq:fn-limits}
\frac{1}{L}\int_{\RR/L\ZZ} f_{j}^2ds \to 1, \qquad \frac{1}{L}\int_{\RR/L\ZZ} (f'_{j})^2ds  \to \frac{\pi^2}{L^2}.     
\end{equation}

On the other hand, letting $\lambda = \lambda^1_P(M_{a,b,L}) > 0$ denote the first $P$-antisymmetric curl eigenvalue on $M_{a,b,L}$, then by Proposition \ref{prop:decomp} and lemmas \ref{lem:complex-decomp} and \ref{lem:new-decomp}, there exist $m \in \ZZ \setminus \{0\}$, $\ell \in (2\pi/L)\ZZ$ and $B$ such that
\begin{equation*}
B \neq 0, \quad \curl B = \lambda B, \quad \cL_T B = i\ell B,\quad \cL_P B = imB, \quad B \cdot n = 0.
\end{equation*}

On the other hand, for $j \in \NN$, let $A_j \coloneqq  A_{f_j}$ and $B_j \coloneqq  B_{f_j}$,
as defined in Lemma \ref{lem:calculations-for-modded-field}. Then, we have that
\begin{equation}\label{eq:Bn-is-admissible}
\Div B_j = 0, \quad \cL_P B_j = i m B_j, \quad B_j \cdot n = 0,
\end{equation}
and by Lemma \ref{lem:vector-pot-and-antisymm-lamda-ineq},
\begin{align}
\cH_\CC(B_j) &= \frac{1}{L}\int_{\RR/L\ZZ} f^2_j ds~\cH_\CC(B) = \frac{1}{\lambda} \frac{1}{L}\int_{\RR/L\ZZ} f^2_j ds~\|B\|_{L^2}^2,\label{eq:helicity-of-Bn}\\
\|B_j\|_{L^2}^2 &\leq \left(\frac{b^2}{m^2L}\int_{\RR/L\ZZ}(f'_j)^2ds + \frac{1}{L}\int_{\RR/L\ZZ} f^2_j ds\right)\|B\|_{L^2}^2.\label{eq:energy-of-Bn}
\end{align}

Now, from the definition of $f_j$, $\supp A_j \subset \pi(U_{a,b,L})$,
and therefore $\supp B_j \subset  \pi(U_{a,b,L})$.    
Since $U_{a,b,L} \subset \cR$, we can then uniquely lift $A_j$ to $\hat{A}_j$ and $B_j$ to $\hat{B}_j$, vector fields on $\cR$ with
$\supp \hat{A}_j \subset U_{a,b,L}$ and $\supp \hat{B}_j \subset U_{a,b,L}$.    

By Equation \eqref{eq:Bn-is-admissible}, it follows that $\hat{B}_j$ satisfies on $\cR$ that
\begin{equation}\label{eq:Bn-hat-is-admissible}
\Div \hat{B}_j = 0, \quad \cL_K \hat{B}_j = i m \hat{B}_j, \quad \hat{B}_j \cdot n = 0.
\end{equation}
Moreover, we still have that $\curl \hat{A}_j = \hat{B}_j$. Therefore, 
\begin{equation}\label{eq:equality-of-hat-and-nonhat}
\begin{split}
\cH_\CC(\hat{B}_j) &= \cH_\CC(B_j),\\
\|\hat{B}_j\|_{L^2}^2 &= \|B_j\|_{L^2}^2. 
\end{split}
\end{equation}

It is clear from Equation \eqref{eq:helicity-of-Bn}, that $\cH(\hat{B}_j) > 0$. Moreover, because $\hat{B}_j$ satisfies Equation \eqref{eq:Bn-hat-is-admissible}, by Lemma \ref{lem:vector-pot-and-antisymm-lamda-ineq} we can write that,
\begin{equation*}
\frac{1}{\lambda^*_1(\cR)} \geq \frac{\cH_\CC(\hat{B}_j)}{\|\hat{B}_j\|_{L^2}^2} = \frac{\cH_\CC(B_j)}{\|B_j\|_{L^2}^2},
\end{equation*}
and hence, together with equations \eqref{eq:helicity-of-Bn} and \eqref{eq:energy-of-Bn}, we conclude that, for $j \in \NN$
\begin{equation*}
\lambda^*_1(\cR) \leq \frac{\|B_j\|_{L^2}^2}{\cH_\CC(B_j)} \leq \frac{\frac{b^2}{m^2L}\int_{\RR/L\ZZ}(f'_j)^2ds + \frac{1}{L}\int_{\RR/L\ZZ} f^2_j ds}{\frac{1}{\lambda} \frac{1}{L}\int_{\RR/L\ZZ} f^2_j ds}.
\end{equation*}
The lemma then follows taking the limit $j\to\infty$ of the right hand side with Equation \eqref{eq:fn-limits}
\begin{equation*}
\lambda^*_1(\cR) \leq \frac{\frac{b^2}{m^2}\frac{\pi^2}{L^2} + 1}{1/\lambda} \leq \left(\frac{\pi^2b^2}{L^2} + 1\right)\lambda = \left(\frac{\pi^2b^2}{L^2} + 1\right)\lambda^1_P(M_{a,b,L}).
\end{equation*}

\section{Proof of Lemma \ref{lem:existence-of-small-eigenvalue}}\label{sec:analysis-annulus-anti-sym}

For $0 < a < b$ and $L > 0$, we consider the Riemannian 3-manifold with boundary $M_{a,b,L}$ as in Section \ref{sec:complex-antisym-fields} and inherit the notation from there. The proof of Lemma \ref{lem:existence-of-small-eigenvalue} is done through analysing the appropriate eigenvalue equation. Specifically, from Lemma \ref{lem:eigen-equations} in Appendix \ref{app:flat-tubes}, given $\ell \in 2\pi/L \ZZ$ and $m \in \ZZ$, there exist nontrivial $B$ and $\lambda \neq 0$ with
\begin{align*}
B \in E_{\CC,\lambda,\ell,m}(M_{a,b,L}),
\end{align*}
if and only if $\lambda^2 > \ell^2$ and, setting $\mu \coloneqq \sqrt{\lambda^2-\ell^2}$, $\lambda$ satisfies
\begin{equation*}
    \begin{vmatrix}
    \frac{\lambda m}{a} J_m(\mu a) + \ell \mu J_m'(\mu a) &  \frac{\lambda m}{a} Y_m(\mu a) + \ell \mu Y_m'(\mu a)\\
    \frac{\lambda m}{b} J_m(\mu b) + \ell \mu J_m'(\mu b) & \frac{\lambda m}{b} Y_m(\mu b) + \ell \mu Y_m'(\mu b)
    \end{vmatrix} = 0,
\end{equation*}
which, thanks to Equation \eqref{eq:bessel-manipulation} in Section \ref{sec:analysis-disk-anti-sym}, is in turn equivalent to
\begin{equation}\label{eq:eigen-value-eq-annulus-reduced}
    \begin{vmatrix}
    (\lambda+\ell) J_{m-1}(\mu a) + (\lambda-\ell) J_{m+1}(\mu a) & (\lambda+\ell) Y_{m-1}(\mu a) + (\lambda-\ell) Y_{m+1}(\mu a)\\
    (\lambda+\ell) J_{m-1}(\mu b) + (\lambda-\ell) J_{m+1}(\mu b) & (\lambda+\ell) Y_{m-1}(\mu b) + (\lambda-\ell) Y_{m+1}(\mu b)
    \end{vmatrix} = 0.
\end{equation}

We are specifically interested in the case where $m=1$. For fixed $b > 0$ and for $0 < r < 1$, we define a function $g_{r} : (0,b) \to \RR$ as follows. Firstly, for $a \in (0,b)$, set
\begin{equation}\label{eq:ansatz}
\begin{gathered}
\lambda_a \coloneqq \frac{\pi}{b-a}, \qquad \ell_{r,a} \coloneqq r\lambda_a,\\
\mu_{r,a} \coloneqq \sqrt{\lambda_a^2-\ell_{r,a}^2} = \sqrt{1-r^2}\frac{\pi}{b-a}.
\end{gathered}
\end{equation}
Then,
\begin{align*}
&\begin{vmatrix}
    (\lambda_a+\ell_{r,a}) J_0(\mu_{r,a} a) + (\lambda_{a}-\ell_{r,a}) J_2(\mu_{r,a} a) & (\lambda_a+\ell_{r,a}) Y_0(\mu_{r,a} a) + (\lambda_{a}-\ell_{r,a}) Y_2(\mu_{r,a} a)\\
    (\lambda_a+\ell_{r,a}) J_0(\mu_{r,a} b) + (\lambda_{a}-\ell_{r,a}) J_2(\mu_{r,a} b) & (\lambda_a+\ell_{r,a}) Y_0(\mu_{r,a} b) + (\lambda_{a}-\ell_{r,a}) Y_2(\mu_{r,a} b)
\end{vmatrix}\\
&= 
\lambda_a^2
\begin{vmatrix}
    (1+r) J_0(\mu_{r,a} a) + (1-r) J_2(\mu_{r,a} a) & (1+r) Y_0(\mu_{r,a} a) + (1-r) Y_2(\mu_{r,a} a)\\
    (1+r) J_0(\mu_{r,a} b) + (1-r) J_2(\mu_{r,a} b) & (1+r) Y_0(\mu_{r,a} b) + (1-r) Y_2(\mu_{r,a} b)
\end{vmatrix}.
\end{align*}
Hence, we define $g_{r} : (0,b) \to \RR$ by
\begin{equation}\label{eq:g-def}
g_r(a) \coloneqq 
\begin{vmatrix}
    (1+r) J_0(\mu_{r,a} a) + (1-r) J_2(\mu_{r,a} a) & (1+r) Y_0(\mu_{r,a} a) + (1-r) Y_2(\mu_{r,a} a)\\
    (1+r) J_0(\mu_{r,a} b) + (1-r) J_2(\mu_{r,a} b) & (1+r) Y_0(\mu_{r,a} b) + (1-r) Y_2(\mu_{r,a} b)
\end{vmatrix}.
\end{equation}
We now establish two helpful lemmas to analyse $g_r$ and its zeros.

\begin{lemma}\label{lem:g-is-eventually-positive}
For any $0<r<1$, the function $g_{r}$ is eventually positive, in the sense that, there exists $a_1 \in (0,b)$ such that $g_{r}|_{(a_1,b)} > 0$.
\end{lemma}

\begin{proof}
We will use well-known asymptotic expressions for Bessel functions with large argument. For fixed $\alpha \in \RR$, as $x \to +\infty$,
\begin{align*}
J_\alpha(x) &= \sqrt{\frac{2}{\pi x}}\cos\left(x-\frac{1}{2}\alpha \pi - \frac{1}{4}\pi\right) + \cO(x^{-3/2})\\
Y_\alpha(x) &= \sqrt{\frac{2}{\pi x}}\sin\left(x-\frac{1}{2}\alpha \pi - \frac{1}{4}\pi\right) + \cO(x^{-3/2}).
\end{align*}
Therefore,
\begin{align*}
(1+r) J_{0}(x) + (1-r) J_{2}(x)=2r\sqrt{\frac{2}{\pi x}}\cos\left(x - \frac{1}{4}\pi\right) + \cO(x^{-3/2}),
\end{align*}
and
\begin{align*}
(1+r) Y_{0}(x) + (1-r) Y_{2}(x)= 2r\sqrt{\frac{2}{\pi x}}\sin\left(x - \frac{1}{4}\pi\right) + \cO(x^{-3/2}).
\end{align*}
Hence, because as $a \to b^{-}$,
\begin{equation*}
\mu_{r,a} \to +\infty,
\end{equation*}
after straightforward computations we can write that
\begin{align*}
g_r(a) = \frac{8r^2(b-a)}{\pi^2\sqrt{1-r^2}\sqrt{ab}}\sin(\pi\sqrt{1-r^2})+\cO((b-a)^{3/2}).
\end{align*}
We see then that there exists the limit
\begin{equation*}
\lim_{a\to b^{-}} \frac{g_r(a)}{b-a} = \frac{8r^2}{\pi^2\sqrt{1-r^2}b}\sin(\pi\sqrt{1-r^2}), 
\end{equation*}
which is positive for $r\in(0,1)$, which proves the claim.
\end{proof}

\begin{lemma}\label{lem:g-is-initially-negative}
There exists $0<r<1$ such that the function $g_{r}$ is initially negative, in the sense that, there exists $a_0 \in (0,b)$ such that $g_{r}|_{(0,a_0)} < 0$.
\end{lemma}

\begin{proof}
We let $0<r<1$ be arbitrary for now. Recall that $g$ is given by
\begin{align*}
g_r(a) =
\begin{vmatrix}
    g_{r,11}(a) & g_{r,12}(a)\\
    g_{r,21}(a) & g_{r,22}(a)
\end{vmatrix},
\end{align*}
with $g_{r,ij}(a)$ denoting entry $(i,j)$ of the matrix whose determinant defines $g_r(a)$ in Equation \eqref{eq:g-def}. Note that, as $a \to 0^+$,
\begin{align*}
g_{r,11}(a) &\to (1+r) J_0(0) + (1-r) J_2(0) = 1+r,\\
g_{r,21}(a) &\to (1+r) J_0(\sqrt{1-r^2}\pi)+(1-r) J_2(\sqrt{1-r^2}\pi),\\
g_{r,22}(a) &\to (1+r) Y_0(\sqrt{1-r^2}\pi) + (1-r) Y_2(\sqrt{1-r^2}\pi). 
\end{align*}
Recalling that for $m\geq 0$, $Y_m(x) \to -\infty$ as $x \to 0^+$, we deduce that as $a \to 0^+$,
\begin{equation*}
g_{r,12}(a) \to -\infty.
\end{equation*}
Hence, if $r_0\in (0,1)$ satisfies that
\begin{equation}\label{eq:r0}
(1+r_0) J_0(\sqrt{1-r_0^2}\pi)+(1-r_0) J_2(\sqrt{1-r_0^2}\pi) < 0,    
\end{equation}
then, as $a \to 0^+$,
\begin{equation*}
g_{r_0}(a) = g_{r_0,11}(a)g_{r_0,22}(a) - g_{r_0,12}(a)g_{r_0,21}(a) \to -\infty, 
\end{equation*}
so that, for such $r_0$, $g_{r_0}$ is initially negative. We prove in Appendix \ref{app:verification-taylor} with Taylor series that if one sets $r_0 = \sqrt{1-\frac{s^2}{\pi^2}}\in (0,1)$ with $s=\frac{287}{100}$,
then \eqref{eq:r0} is satisfied, and the lemma follows.
\end{proof}

We now continue with the Lemma \ref{lem:existence-of-small-eigenvalue}. By lemmas \ref{lem:g-is-eventually-positive}, \ref{lem:g-is-initially-negative}, we may fix $0 < r < 1$ and $0 < a < b$ such that
\begin{equation*}
g_r(a) = 0.    
\end{equation*}
Now, set $\lambda = \lambda_a = \pi/(b-a)$ and $\ell = \ell_{r,a}$ as they are defined in Equation \eqref{eq:ansatz}, and
\begin{equation*}
L = \frac{2\pi}{\ell} = \frac{2(b-a)}{r}.   
\end{equation*}
For $n \in \NN$, set $L_n = n L$. Then, by construction, 
$\ell = \frac{2\pi n}{L_n}\in (2\pi/L_n)\ZZ$, and by Equation~\eqref{eq:ansatz}, we also have that $\ell^2< \lambda^2$. Moreover, $g_r(a) = 0$ implies that Equation \eqref{eq:eigen-value-eq-annulus-reduced} is satisfied by $\lambda$, $\ell$ and $\mu = \sqrt{\lambda^2-\ell^2}$ when $m = 1$. Therefore, for $n \in \NN$,
$E_{\CC,\lambda,\ell,1}(M_{a,b,L_n}) \neq \{0\}$, which shows that $\pi/(b-a)$ is a $P$-antisymmetric curl eigenvalue on $M_{a,b,L_n}$. Lemma~\ref{lem:existence-of-small-eigenvalue} follows.

\appendix

\section{Proof of Lemma \ref{lem:basic-G-S-ineq}}\label{app:Grad-Shafranov}

Let $U \subset \RR^2$ be a bounded domain with $\overline{U} \subset \RR^+ \times \RR$. Here we consider the (Dirichlet) eigenvalues of the Grad-Shafranov operator
\begin{equation*}
L = r \Div(1/r \nabla) = \Delta - \nabla \ln(r) \cdot \nabla,     
\end{equation*}
with $(r,z) \in U$ the usual Cartesian coordinates. For $u,v \in C^\infty_c(U)$,
\begin{equation*}
-(L u,v)_{L^2,r} = B_r[u,v],     
\end{equation*}
where $(\cdot,\cdot)_{L^2,r}$ is the weighted inner product  given by
\begin{equation*}
(u,v)_{L^2,r} = \int_U uv/r~drdz,   
\end{equation*}
and $B_r : H^1_0(U) \times H^1_0(U) \to \RR$ is the bilinear form given by
\begin{equation*}
B_r[u,v] = \int_U (\nabla u \cdot \nabla v)/r~drdz.
\end{equation*}
The eigenfunctions of $L$, that is, the functions $u \in H^1_0(U)$ satisfying (weakly) that
\begin{equation*}
L u = -\lambda u \text { on } U,
\end{equation*}
can be used to form an $(\cdot,\cdot)_{L^2,r}$-orthonormal basis of $L^2(U)$ and the first eigenvalue $\lambda^{\text{G-S}}_1(U)$ satisfies that
\begin{equation}\label{eq:min-char-of-eigenvalue}
\lambda^{\text{G-S}}_1(U) = \min_{0 \neq u \in H^1_0(U)} \frac{B_r[u,u]}{\|u\|^2_{L^2,r}} = \inf_{0 \neq u \in C^\infty_c(U)} \frac{B_r[u,u]}{\|u\|^2_{L^2,r}}.
\end{equation}
This makes it obvious that domain monotonicity is satisfied: if $U' \subset \RR^2$ satisfies $\overline{U'} \subset \RR^+ \times \RR$ and $U \subset U'$, then $\lambda^{\text{G-S}}_1(U) \geq \lambda^{\text{G-S}}_1(U')$.

To prove the second inequality in Lemma \ref{lem:basic-G-S-ineq}, note that for any $u \in C^\infty_c(U)$, setting $v_u = r^{-1/2}u \in C^\infty_c(U)$, we have that $\|v_u\|_{L^2} = \|u\|_{L^2,r}$
and after a straightforward computation, we get
\begin{align*}
\int_U |\nabla v_u|^2~drdz = -\int_U f u^2/r~drdz + B_r[u,u],
\end{align*}
with $f(r) \coloneqq \frac{3}{4r^2}$ and therefore
\begin{align*}
\frac{\int_U |\nabla v_u|^2~drdz}{\|v_u\|_{L^2}^2} + \inf_{(r,z) \in U}f(r) \leq \frac{B_r[u,u]}{\|u\|_{L^2,r}^2} \leq \frac{\int_U |\nabla v_u|^2~drdz}{\|v_u\|_{L^2}^2} + \sup_{(r,z) \in U}f(r).
\end{align*}
Taking the infimum over $0 \neq u \in C^\infty_c(U)$ (with Equation \eqref{eq:min-char-of-eigenvalue} in mind), and using the fact that $u \mapsto v_u$ is a linear isomorphism $C^\infty_c(U) \to C^\infty_c(U)$, we get that
\begin{align*}
\lambda^D_1(U) + \inf_{(r,z) \in U}f(r) \leq \lambda^{\text{G-S}}_1(U) \leq \lambda^D_1(U) + \sup_{(r,z) \in U}f(r),
\end{align*}
which proves the lemma.

\section{Functional analysis of the antisymmetric curl spectrum}\label{app:func-analysis-antisym-spec}

Let $M$ be a compact orientable Riemannian 3-manifold with boundary. We first mention some useful spectral theory for the curl in this setting. To do this, define the spaces of vector fields
\begin{align*}
\cB(M) &= \{B \in \cX(M) : \Div B = 0, \, B \cdot n = 0\},\\
\cH_N(M) &= \{B \in \cB(M) : \curl B = 0\},\\
\cB_{\text{FF}}(M) &= \{B \in \cB(M) : (B,\nu)_{L^2} = 0 \text{ for all } \nu \in \cH_N(M)\},
\end{align*}
where $n$ denotes the outward unit normal on $M$. For $\lambda \in \RR \setminus \{0\}$, we also define the space
\begin{equation}\label{eq:def-curl-eigenspace}
E_\lambda(M) =  \{B \in \cX(M) : \curl B = \lambda B,\, B \cdot n = 0\}.
\end{equation}
As is standard, $\cH_N(M)$ has dimension $b^1$, the first Betti number of $M$, and $E_\lambda(M)$ is finite dimensional for $\lambda \in \RR \setminus \{0\}$. Additionally, the following is well-known to experts.

\begin{theorem}\label{thm:flux-free-spectral}
For each $B \in \cB_{\text{FF}}(M)$ there exists a unique $A \in \cB_{\text{FF}}(M)$ such that
\begin{equation*}
\curl A = B.    
\end{equation*}
We let $\curl^{-1} : \cB_{\text{FF}}(M) \to \cB_{\text{FF}}(M)$ denote the corresponding curl inverse operator. The operator $\curl^{-1}$ extends to a compact self-adjoint operator
\begin{equation*}
\curl^{-1} : L^2\cB_{\text{FF}}(M) \to L^2\cB_{\text{FF}}(M)
\end{equation*}
on the $L^2$-completion $L^2\cB_{\text{FF}}(M)$ of $\cB_{\text{FF}}(M)$. The distinct eigenvalues of $\curl^{-1}$ can be written as
\begin{equation*}
\mu_{-1}(M) < \mu_{-2}(M) < \cdots < 0 < \cdots < \mu_2(M) < \mu_1(M), 
\end{equation*}
with $\mu_n(M) \to 0$ as $n \to \pm \infty$. For any $\lambda \in \RR$ and any $B \in L^2\cB_{\text{FF}}(M)$, we have that $B \in \cB_{\text{FF}}(M)$ and
\begin{equation}\label{eq:flux-free-curl-eq}
\curl B = \lambda B,    
\end{equation}
if and only if $\lambda \neq 0$ and
\begin{equation*}
\curl^{-1} B = \frac{1}{\lambda} B.    
\end{equation*}
In particular, the distinct \emph{flux-free curl eigenvalues}, that is, the $\lambda \in \RR$ for which Equation \eqref{eq:flux-free-curl-eq} has a non-trivial solution $B \in \cB_{\text{FF}}(M)$, can be written as
\begin{equation*}
\cdots < \lambda_{-2}^{\text{FF}}(M) < \lambda_{-1}^{\text{FF}}(M) < 0 < \lambda_1^{\text{FF}}(M) < \lambda_2^{\text{FF}}(M) < \cdots, 
\end{equation*}
with $\lambda_n^{\text{FF}}(M) \to \pm \infty$ as $n \to \pm \infty$.
\end{theorem}

We recall here for later use that helicity on $L^2\cB_{\text{FF}}(M)$ is the continuous functional $\cH : L^2\cB_{\text{FF}}(M) \to \RR$ given by
\begin{equation*}
\cH(B) = (B,\curl^{-1}B)_{L^2}.
\end{equation*}
An important property of $\cH$ is that, for any $B \in L^2\cB_{\text{FF}}(M)$ and any $A \in H^1\cX(M)$ with $\curl A = B$, one has that
\begin{equation*}
\cH(B) = (B,A)_{L^2}.    
\end{equation*}

Assume henceforth that $M$ admits a Killing field $K$ tangent to $\partial M$. The following lemma is elementary:

\begin{lemma}\label{lem:antisym-of-K}
The Lie derivative $\cL_K : \cX(M) \to \cX(M)$ is antisymmetric in the sense that
\begin{equation*}
(\cL_K X,Y)_{L^2} = -(X,\cL_K Y)_{L^2}
\end{equation*}
for $X,Y \in \cX(M)$. 
\end{lemma}

We can now prove Proposition \ref{prop:decomp}  from Section \ref{sec:intro-antisym-spec}.

\begin{proof}[Proof of Proposition \ref{prop:decomp}]
Because $K$ is a Killing field tangent to $\partial M$, we have that $\cL_K(E_\lambda(M)) \subset E_{\lambda}(M)$. Thus, by Lemma \ref{lem:antisym-of-K} $\cL_K$ descends to an antisymmetric linear operator $L$ on the finite dimensional vector space $E_{\lambda}(M)$ (equipped with the $L^2$ inner product). The antisymmetry of $L : E_{\lambda}(M) \to E_{\lambda}(M)$ means that $L^2 = L \circ L : E_{\lambda}(M) \to E_{\lambda}(M)$ is a symmetric operator. Hence, from the Spectral Theorem, we obtain 
\begin{equation*}
E_{\lambda}(M) = \bigoplus_{i=1}^{n'} A_i
\end{equation*}
for some $n' \in \NN$, where $A_i$ are the distinct (mutually orthogonal) eigenspaces of $L^2$. For any $0 \neq X \in \cX(M)$ and $\gamma \in \RR$ such that
\begin{equation*}
\cL_K^2 X = \gamma X,    
\end{equation*}
we have by Lemma \ref{lem:antisym-of-K} that
\begin{equation*}
\gamma(X,X)_{L^2} = (\cL_K^2 X,X)_{L^2} = -(\cL_K X,\cL_K X)_{L^2}  
\end{equation*}
which implies that, if $\gamma \neq 0$, then $\gamma < 0$, and if $\gamma = 0$, then $\cL_K X = 0$. Thus, for each $i \in \{1,...,m\}$, either
\begin{equation*}
A_i = \ker L = E_\lambda(M) \cap (\ker \cL_K),    
\end{equation*}
or
\begin{equation*}
A_i = E_{\lambda,K,k}(M),   
\end{equation*}
for some $k > 0$. Hence, we are done.
\end{proof}

The final step towards developing the functional analysis for $K$-antisymmetric curl is the following.

\begin{lemma}\label{lem:compatibility-of-antisym}
We have the inclusions
\begin{align}
\cB_K(M) &\subset \cB_{\text{FF}}(M), \label{eq:anti-symm-implies-FF}\\
\curl^{-1}(\cB_K(M)) &\subset \cB_K(M),\label{eq:curl-inv-respects-antisym}
\end{align}
where $\curl^{-1} : \cB_{\text{FF}}(M) \to \cB_{\text{FF}}(M)$ denotes the inverse curl operator on $\cB_{\text{FF}}(M)$. 
\end{lemma}

\begin{proof}
We will use repeatedly the well-known fact that for any Neumann harmonic vector field $\nu \in \cH_N(M)$, $\cL_K \nu = 0$. Then, Equation~\eqref{eq:anti-symm-implies-FF} follows immediately from Lemma \ref{lem:antisym-of-K}.

Concerning \eqref{eq:curl-inv-respects-antisym}, if $B \in \cB_K(M)$, then write $B = \cL_K \tilde{B}$ for some $\tilde{B} \in \cB(M)$. We may assume $\tilde{B} \in \cB_{\text{FF}}(M)$ from the first paragraph together with the fact that $\cH_N(M)$ is finite dimensional. Setting $\tilde{A} = \curl^{-1} \tilde{B}\in \cB_{\text{FF}}(M) \subset \cB(M)$, we have $\curl \tilde{A} = \tilde{B}$. Considering $A = \cL_K \tilde{A} \in \cB_K(M)$, because $K$ is Killing,
$\curl A = B$,
and hence, by~\eqref{eq:anti-symm-implies-FF}, $A = \curl^{-1} B$, and the lemma follows.
\end{proof}

Finally, we state an immediate corollary of Lemma \ref{lem:compatibility-of-antisym}, Proposition \ref{prop:decomp}, and Theorem \ref{thm:flux-free-spectral}:

\begin{corollary}\label{cor:antisym-spectral}
Let $L^2\cB_K(M)$ denote the $L^2$-completion of $\cB_K(M)$. Then, one has the natural inclusion
\begin{equation*}
L^2\cB_K(M) \subset L^2\cB_{\text{FF}}(M),    
\end{equation*}
with respect to which $L^2\cB_K(M)$ is a closed vector subspace of $L^2\cB_{\text{FF}}(M)$. Moreover, $\curl^{-1}$ descends to a compact self-adjoint operator
\begin{equation*}
\curl^{-1}_K : L^2\cB_K(M) \to L^2\cB_K(M).
\end{equation*}
Let $\lambda \in \RR$. Then for any  $B \in L^2\cB_K(M)$, we have that $B \in \cB_K(M)$ and
\begin{equation*}
\curl B = \lambda B,    
\end{equation*}
if and only if $\lambda \neq 0$ and
\begin{equation*}
\curl^{-1}_K B = \frac{1}{\lambda} B.    
\end{equation*}
In particular, the distinct $K$-antisymmetric curl eigenvalues can be written as
\begin{equation*}
\cdots < \lambda_{-2}^K(M) < \lambda_{-1}^K(M) < 0 < \lambda_1^K(M) < \lambda_2^K(M) < \cdots, 
\end{equation*}
with $\lambda_n^K(M) \to \pm \infty$ as $n \to \pm \infty$ and, for any $  B \in L^2\cB_K(M)$, we have that
\begin{equation*}
\frac{1}{\lambda^K_{-1}(M)} \leq \frac{\cH(B)}{\|B\|_{L^2}^2} \leq \frac{1}{\lambda^K_1(M)},    
\end{equation*}
with the lower bound being saturated if and only if $B \in E_{\lambda^K_{-1}}(M) \cap \cB_K(M)$, and the upper bound being saturated if and only if $B \in E_{\lambda^K_1}(M) \cap \cB_K(M)$.
\end{corollary}

\section{The curl spectrum on flat tubes}\label{app:flat-tubes}

To begin, let $(M,g)$ be a compact  Riemannian 3-manifold with boundary. Let $K$ be a non-trivial Killing field on $(M,g)$ tangent to $\partial M$. For $Z \in \cX(M,\CC)$, setting $h_Z\coloneqq g(Z,K)$, the following identity is elementary:

\begin{lemma}\label{lem:Lie-Killing-id}
Let $Z \in \cX(M,\CC)$. Then,
\begin{equation*}
\cL_K Z = (\curl Z) \times K + \nabla h_Z.
\end{equation*}
\end{lemma}

\begin{proof}
First note that, by Cartan's formula,
\begin{align*}
\cL_K (Z^{\flat}) = i_K dZ^{\flat} + d i_K Z^{\flat} &= i_K i_{\curl Z} \mu + dh_Z\\
&= ((\curl Z) \times K)^\flat + (\nabla h_Z)^{\flat}. 
\end{align*}
On the other hand, because $K$ is Killing, $\cL_K (Z^{\flat}) = (\cL_K Z)^{\flat}$. The result follows.
\end{proof}

The following proposition shows that a curl eigenfield cannot be pointwise orthogonal to a Killing field.

\begin{proposition}\label{prop:injectivity-of-scalar-projection}
Let $\lambda \in \RR \setminus \{0\}$ and $B\in E_{\CC,\lambda}(M)$. Then $g(B,K)=0$ if and only if $B=0$.
\end{proposition}

\begin{proof}
Since $K$ is real-valued, it is enough to work with $B \in E_{\lambda}(M)$. Assume that $h_B=g(B,K)=0$. By Lemma \ref{lem:Lie-Killing-id}, we then have that
\begin{equation*}
\cL_K B = \lambda B \times K.
\end{equation*}
Because $B$ and $K$ are tangent to $\partial M$, the left hand side is tangent to $\partial M$ while the right hand side is perpendicular to $\partial M$. Thus
\begin{equation*}
(\cL_K B) |_{\partial M} = 0.
\end{equation*}
However, $\cL_K B \in E_{\lambda}(M)$ and so
\begin{equation*}
\lambda B \times K = \cL_K B = 0  
\end{equation*}
on $M$ because any Beltrami field that vanishes on the boundary of the manifold is zero everywhere, see e.g.~\cite[Lemma 2.1]{Gerner2021typical} and \cite[Proposition 3.8]{EPS23}.

In summary, $g(B,K) = 0$ and $B \times K = 0$, so $B=0$
on the non-empty open subset $V = \{x \in M : K|_x \neq 0\}$, and $B=0$ on the whole $M$ by the unique continuation property.
\end{proof}

Set
\begin{equation}\label{eq:def-of-sigma}
\sigma = g(K,K).    
\end{equation}
The following is an approach for expressing a field $B \in E_{\CC,\lambda,k}(M)$ ($k \in \RR$, $\lambda \in \RR \setminus \{0\}$) in terms of $h_B$.

\begin{proposition}\label{prop:determining-expression}
Let $\lambda \in \RR \setminus \{0\}$ and $k \in \RR$. If $ B \in E_{\CC,\lambda,k}(M)$, then $\cL_K h_B = ik h_B$ and one has
\begin{align}
(\lambda^2 \sigma - k^2) B &= \lambda^2 h_B K + \lambda \nabla h_B \times K + ik \nabla h_B\label{eq:main-expression-for-B},\\
\lambda^2\max \sigma - k^2 &> 0\label{eq:antisym-lambda-lower-bound}.
\end{align}
\end{proposition}

\begin{proof}
That $\cL_K h_B = ik h_B$ easily follows using that $K$ is Killing and $\cL_K B = ikB$.
For the second claim, by Lemma \ref{lem:Lie-Killing-id}, $\cL_K B = \lambda B \times K + \nabla h_B$, and
hence, 
\begin{equation}\label{eq:basic-expression}
ik B = \lambda B \times K + \nabla h_B.
\end{equation}
Now, crossing Equation~\eqref{eq:basic-expression} gives
\begin{equation}\label{eq:alt-expression-for-cross}
ik B \times K =  \lambda (h_B K- \sigma B) + \nabla h_B \times K.
\end{equation}
Substituting Equation~\eqref{eq:alt-expression-for-cross} into  Equation~\eqref{eq:basic-expression} yields
\begin{equation*}
-k^2 B = \lambda^2 (h_B K- \sigma B) + \lambda \nabla h_B \times K + ik \nabla h_B.
\end{equation*}
Equation~\eqref{eq:main-expression-for-B} then follows by rearrangement. Concerning \eqref{eq:antisym-lambda-lower-bound}, we recall from Lemma \ref{lem:vector-pot-and-antisymm-lamda-ineq} in Section \ref{sec:complex-antisym-fields} that
\begin{equation*}
\cH_\CC(B) = \frac{i}{k}(B,B\times K)_{L^2}.    
\end{equation*}
In particular,
\begin{align*}
|\cH_\CC(B)|^2 \leq \frac{1}{k^2}\|B\|_{L^2}^2\|B\times K\|_{L^2}^2.
\end{align*}
Moreover, 
\begin{equation*}
\|B\times K\|_{L^2}^2 = \int_M (\sigma g_\CC(B,B) - h_B\overline{h_B}) \mu \leq \max \sigma \|B\|_{L^2}^2 - \|h_B\|_{L^2}^2,
\end{equation*}
so that
\begin{equation}\label{eq:helicity-estimate}
|\cH_\CC(B)|^2 \leq \frac{\max \sigma}{k^2}\|B\|_{L^2}^4 - \frac{1}{k^2}\|h_B\|_{L^2}^2\|B\|_{L^2}^2.
\end{equation}
Additionally, because $B \neq 0$, $h_B \neq 0$ by Proposition~\ref{prop:injectivity-of-scalar-projection} and so estimate \eqref{eq:helicity-estimate} implies that
\begin{equation*}
\frac{1}{\lambda^2}\|B\|_{L^2}^4 < \frac{\max \sigma}{k^2}\|B\|_{L^2}^4,
\end{equation*}
and hence~\eqref{eq:antisym-lambda-lower-bound}.
\end{proof}

Let $0 < a < b$ and $L > 0$. Henceforth assume that
\begin{equation*}
M = \Sigma \times \RR/L\ZZ,    
\end{equation*}
where $\Sigma \in \{D_a,A_{a,b}\}$ and
\begin{align*}
D_a &= \{(x,y) \in \RR^2 : x^2+y^2 \leq a^2\},\\  
A_{a,b} &= \{(x,y) \in \RR^2 : a^2 \leq x^2+y^2 \leq b^2\}. 
\end{align*}
Let $(x,y)\in \Sigma$ and $s : M \to \RR/L\ZZ$ denote the Cartesian coordinates on $M$. Consider the flat metric $g$ and associated Killing fields $T$ and $P$ given by
\begin{equation*}
g = dx^2 + dy^2 + ds^2, \quad T = \partial_s, \quad P = -y\partial_x + x\partial_y.
\end{equation*}

Now, for $\lambda \in \RR \setminus \{0\}$ and $\ell \in (2\pi/L)\ZZ$, we set
\begin{equation*}
\sigma_{\lambda,\ell} \coloneqq  \lambda^2-\ell^2.    
\end{equation*}
Then, by Lemma \ref{lem:vector-pot-and-antisymm-lamda-ineq}, if $E_{\CC,\lambda,\ell,m}(M) \neq \{0\}$ for some $m \in \ZZ$, then because $g(T,T) = 1$, we have that $\lambda^2 > \ell^2$.
Hence, we may assume that $\sigma_{\lambda,\ell} > 0$ when characterising the spaces $E_{\CC,\lambda,\ell,m}(M)$. We complete the characterisation in two steps. In what follows, we set $h_B\coloneqq g(B,T)$.

\begin{lemma}\label{lem:abstract-BVP-for-h-radially-sym}
 $B \in E_{\CC,\lambda,\ell,m}(M)$ if and only if there exists $h \in C^{\infty}(M,\CC)$ such that
\begin{equation*}
B = \lambda^2 h T + \lambda \nabla h \times T + i\ell \nabla h,    
\end{equation*}
where $h$ satisfies the equations
\begin{equation}\label{eq:h-differential-equations-for-flat-tube}
\Delta h = -\lambda^2 h,\qquad T(h) = i\ell h,\qquad P(h) = im h,
\end{equation}
together with the boundary condition
\begin{equation}\label{eq:h-boundary-condition-disk}
\frac{\lambda m}{a} h|_{\partial M} + \ell \frac{\partial h}{\partial n} = 0,
\end{equation}
if $\Sigma = D_a$, and together with the boundary condition
\begin{equation}\label{eq:h-boundary-condition-annulus}
\begin{split}
-\frac{\lambda m}{a} h  + \ell \frac{\partial h}{\partial n} &= 0 \text{ on } \partial_1 M,\\
\frac{\lambda m}{b} h + \ell \frac{\partial h}{\partial n} &= 0 \text{ on } \partial_2 M,
\end{split}
\end{equation}
if $\Sigma = A_{a,b}$, where 
\begin{align*}
\partial_1 M &= \{(x,y) \in \RR^2 : x^2+y^2 = a^2\} \times \RR/L\ZZ,\\
\partial_2 M &= \{(x,y) \in \RR^2 : x^2+y^2 = b^2\} \times \RR/L\ZZ.
\end{align*}
\end{lemma}

\begin{proof}
First, if $B \in E_{\CC,\lambda,\ell,m}(M)$, then by Proposition \ref{prop:determining-expression},
\begin{align*}
\cL_T h_B &= i\ell h_B,\\
\sigma_{\lambda,\ell} B &= \lambda^2 h_B T + \lambda \nabla h_B \times T + i\ell \nabla h_B,
\end{align*}
but also, because $P$ is Killing and $[P,T] = 0$,
\begin{equation*}
\cL_P h_B = im h_B.
\end{equation*}
Hence, setting $h \coloneqq  h_B/\sigma_{\lambda,\ell}$, we can say that $B$ and $h$ satisfy
\begin{equation}\label{eq:anzatz}
\begin{split}
B &= \lambda^2 h T + \lambda \nabla h \times T + i\ell \nabla h,\\
\cL_T h &= i\ell h,\\
\cL_P h &= im h.
\end{split}    
\end{equation}

Now, let $B \in \cX(M,\CC)$ and $h \in C^{\infty}(M,\CC)$ be such that Equation \eqref{eq:anzatz} is satisfied. Then, because $T$ is Killing, we easily see that
\begin{equation*}
\cL_T B =  i\ell B.
\end{equation*}
Likewise, because $P$ is Killing and $[P,T] = 0$,
\begin{equation*}
\cL_P B = im B.
\end{equation*}
Now, because $\curl T = 0$ and $T$ is Killing, a straightforward computation shows
\begin{align*}
\curl B = -\lambda \Delta h T + \lambda^2 \nabla h \times T + i\ell \lambda \nabla h.
\end{align*}
Hence, we see that
\begin{equation*}
\curl B = \lambda B \iff \Delta h = -\lambda^2 h.
\end{equation*}
Using the fact that $T \cdot n = 0$, we have from Equation \eqref{eq:anzatz} that
\begin{equation*}
B \cdot n = \lambda (\nabla h \times T)\cdot n + i\ell \frac{\partial h}{\partial n}.    
\end{equation*}

Now, if $\Sigma = D_a$, then the outward unit normal is given as
\begin{equation*}
n = \frac{x\partial_x + y \partial_y}{a} \bigg|_{\partial M}.    
\end{equation*}
Hence,
\begin{align*}
(\nabla h \times T)\cdot n = (T \times n) \cdot \nabla h = \frac{1}{a} P \cdot \nabla h |_{\partial M} =\frac{im}{a} h|_{\partial M},  
\end{align*}
and so Equation $B\cdot n=0$ is equivalent to Equation \eqref{eq:h-boundary-condition-disk} in this setting. 

If $\Sigma = A_{a,b}$, then the outward unit normal is given as
\begin{align*}
n|_{\partial_1 M} &= - \frac{x\partial_x + y \partial_y}{a} \bigg|_{\partial_1 M}, \\
n|_{\partial_2 M} &= 
\frac{x\partial_x + y \partial_y}{b} \bigg|_{\partial_2 M},
\end{align*}
and one similarly computes that
\begin{align*}
(\nabla h \times T)\cdot n|_{\partial_1 M} &= -\frac{im}{a} h|_{\partial_1 M},\\  
(\nabla h \times T)\cdot n|_{\partial_2 M} &= \frac{im}{b} h|_{\partial_2 M},
\end{align*}
so that Equation $B\cdot n=0$ is equivalent to Equation \eqref{eq:h-boundary-condition-annulus}. This proves the result.
\end{proof}

One can explicitly write the form of $h$ solving Equation \eqref{eq:h-differential-equations-for-flat-tube} and the eigenvalue equation coming about from the boundary condition in Lemma \ref{lem:abstract-BVP-for-h-radially-sym}. This is the content of the last lemma of this section.

\begin{lemma}\label{lem:eigen-equations}
Let $\ell \in (2\pi/L)\ZZ$ and $m \in \ZZ$. Let $\lambda \in \RR$ such that $\lambda^2 > \ell^2$. Set $\mu \coloneqq \sqrt{\lambda^2-\ell^2}$. Let $0 \neq h \in C^{\infty}(M,\CC)$. Then, the following holds.
\begin{enumerate}
    \item If $\Sigma = D_a$, then $h$ solves equations \eqref{eq:h-differential-equations-for-flat-tube} and \eqref{eq:h-boundary-condition-disk} if and only if there exists $A \in \CC \setminus \{0\}$ such that
    \begin{equation}\label{eq:form-of-h-disk}
    h = A J_m(\mu r)e^{im\theta}e^{i \ell s},
    \end{equation}
    where $\lambda$, $\ell$, and $m$, additionally satisfy
    \begin{equation}\label{eq:eigenvalue-equation-disk}
    \frac{\lambda m}{a} J_m(\mu a) + \ell \mu J_m'(\mu a) = 0.    
    \end{equation}
    \item If $\Sigma = A_{a,b}$, then $h$ solves equations \eqref{eq:h-differential-equations-for-flat-tube} and \eqref{eq:h-boundary-condition-annulus} if and only if there exists $(A,B) \in \CC^2 \setminus \{(0,0)\}$ such that
    \begin{equation}\label{eq:form-of-h-annulus}
    h = (A J_m(\mu r) + BY_m(\mu r))e^{im\theta}e^{i \ell s},
    \end{equation}
    where $A$, and $B$, satisfy
    \begin{equation}\label{eq:mult-is-zero-annulus}
    \begin{pmatrix}
    \frac{\lambda m}{a} J_m(\mu a) + \ell \mu J_m'(\mu a) &  \frac{\lambda m}{a} Y_m(\mu a) + \ell \mu Y_m'(\mu a)\\
    \frac{\lambda m}{b} J_m(\mu b) + \ell \mu J_m'(\mu b) & \frac{\lambda m}{b} Y_m(\mu b) + \ell \mu Y_m'(\mu b)
    \end{pmatrix}
    \begin{pmatrix}
    A\\
    B
    \end{pmatrix}
    = 0. 
    \end{equation}
   The existence of $(A,B) \in \CC^2 \setminus \{(0,0)\}$ satisfying Equation \eqref{eq:mult-is-zero-annulus} is of course equivalent to $\lambda$, $\ell$ and $m$ satisfying
    \begin{equation}\label{eq:det-is-zero-annulus}
    \begin{vmatrix}
    \frac{\lambda m}{a} J_m(\mu a) + \ell \mu J_m'(\mu a) &  \frac{\lambda m}{a} Y_m(\mu a) + \ell \mu Y_m'(\mu a)\\
    \frac{\lambda m}{b} J_m(\mu b) + \ell \mu J_m'(\mu b) & \frac{\lambda m}{b} Y_m(\mu b) + \ell \mu Y_m'(\mu b)
    \end{vmatrix} = 0.
    \end{equation}
\end{enumerate}
\end{lemma}

\begin{proof}
Using polar coordinates $(r,\theta)$ in $(x,y) \in \Sigma$, $h$ being a solution to Equation \eqref{eq:h-differential-equations-for-flat-tube} is equivalent to
\begin{equation*}
h = f(r)e^{im\theta}e^{i\ell s},    
\end{equation*}
for some smooth $f$ (defined on $[0,a]$ if $\Sigma = D_a$ and defined on $[a,b]$ if $\Sigma = A_{a,b}$) which satisfies the (re-scaled) Bessel equation
\begin{equation}\label{eq:Bessel}
r^2\frac{d^2f}{dr^2} + r\frac{df}{dr} + (r^2(\lambda^2-\ell^2)-m^2)f = 0.
\end{equation}
We now treat the disk and annulus cases individually.

Suppose $\Sigma = D_a$. Then, from the previous discussion, $h \neq 0$ solves equations \eqref{eq:h-differential-equations-for-flat-tube} and \eqref{eq:h-boundary-condition-disk} if and only if there exists $A \in \CC \setminus \{0\}$ such that Equation \eqref{eq:form-of-h-disk} holds. Assuming that Equation \eqref{eq:form-of-h-disk} holds, we compute 
\begin{align*}
\frac{\lambda m}{a} h|_{\partial M} + \ell \frac{\partial h}{\partial n} = A\left(\frac{\lambda m}{a} J_m(\mu a) + \ell \mu J_m'(\mu a)\right)e^{im\theta}e^{i \ell s}.  
\end{align*}
This proves the result in the disk case.

Now suppose $\Sigma = A_{a,b}$. As before, $h \neq 0$ solves equations \eqref{eq:h-differential-equations-for-flat-tube} and \eqref{eq:h-boundary-condition-disk} if and only if there exists $(A,B) \in \CC^2 \setminus \{(0,0)\}$ such that Equation \eqref{eq:form-of-h-annulus} holds. Assuming that Equation \eqref{eq:form-of-h-annulus} holds, letting $f : [a/\mu,b/\mu] \to \RR$ be a smooth function so that
\begin{equation*}
f(\mu r) \coloneqq  A J_m(\mu r) + B Y_m(\mu r),     
\end{equation*}
we compute on $\partial_1 M$ that
\begin{align*}
-\frac{\lambda m}{a} h + \ell \frac{\partial h}{\partial n} = -\left(\frac{\lambda m}{a} f(\mu a) + \ell \mu f'(\mu a)\right)e^{im\theta}e^{i \ell s}.
\end{align*}
Hence, 
\begin{equation*}
-\frac{\lambda m}{a} h  + \ell \frac{\partial h}{\partial n} = 0 \text{ on } \partial_1 M,   
\end{equation*}
is equivalent to
\begin{equation}\label{eq:first-boundary-part-annulus}
A \left(\frac{\lambda m}{a} J_m(\mu a) + \ell \mu J_m'(\mu a)\right) + B \left(\frac{\lambda m}{a} Y_m(\mu a) + \ell \mu Y_m'(\mu a)\right) = 0.
\end{equation}
We compute on $\partial_2 M$ that
\begin{align*}
\frac{\lambda m}{b} h + \ell \frac{\partial h}{\partial n} = \left(\frac{\lambda m}{b} f(\mu b) + \ell \mu f'(\mu b)\right)e^{im\theta}e^{i \ell s}.
\end{align*}
Hence, 
\begin{equation*}
\frac{\lambda m}{b} h  + \ell \frac{\partial h}{\partial n} = 0 \text{ on } \partial_2 M,   
\end{equation*}
is equivalent to
\begin{equation}\label{eq:second-boundary-part-annulus}
A \left(\frac{\lambda m}{b} J_m(\mu b) + \ell \mu J_m'(\mu b)\right) + B \left(\frac{\lambda m}{b} Y_m(\mu b) + \ell \mu Y_m'(\mu b)\right) = 0.
\end{equation}
Hence, in this case, Equation \eqref{eq:h-boundary-condition-annulus} is therefore equivalent to the conjunction of equations \eqref{eq:first-boundary-part-annulus} and \eqref{eq:second-boundary-part-annulus}, which can of course rewritten as Equation \eqref{eq:mult-is-zero-annulus}. This proves the result for the annulus.
\end{proof}

\section{A computer assisted proof with Taylor series}\label{app:verification-taylor}

The purpose of this appendix is to verify the claim made in the proof of Lemma \ref{lem:g-is-initially-negative} in Section \ref{sec:analysis-annulus-anti-sym} which is that, if one sets $r = \sqrt{1-\frac{s^2}{\pi^2}}$ with
$s= \frac{287}{100}$
then
\begin{equation*}
(1+r) J_0(\sqrt{1-r^2}\pi)+(1-r) J_2(\sqrt{1-r^2}\pi) < 0.   
\end{equation*}

\subsection{A Taylor remainder estimate for Bessel functions}

Recall that for $n \in \NN_0$, the Bessel function $J_n$ is given by the power series
\begin{equation*}
J_n(x) = \sum_{m = 0}^\infty \frac{(-1)^m}{m!(m+n)!}\left(\frac{x}{2}\right)^{2m+n}    
\end{equation*}
with infinite radius of convergence. For $M \in \NN $, consider the remainder $R_{M,n}$ given by
\begin{equation*}
R_{M,n}(x) = \sum_{m = M+1}^\infty \frac{(-1)^m}{m!(m+n)!}\left(\frac{x}{2}\right)^{2m+n}. 
\end{equation*}
Then, for $|x| \leq 3$,
\begin{align*}
|R_{M,n}(x)| \leq  \sum_{m = M+1}^\infty \frac{1}{m!(m+n)!}\left(\frac{3}{2}\right)^{2m+n}.     
\end{align*}
For integer $m \geq 2$, defining $\prod_{j=1}^n (m+j)=1$ as an empty product when $n=0$, we have that
\begin{equation*}
\frac{1}{(m+n)!}\left(\frac{3}{2}\right)^{n} = \frac{1}{m!}\frac{(3/2)^n}{\prod_{j=1}^n (m+j)} \leq \frac{1}{m!} \frac{(3/2)^n}{2^n} \leq \frac{1}{m!}.  
\end{equation*}
So, because $M+1\geq 2$,
\begin{equation*}
|R_{M,n}(x)| \leq \sum_{m = M+1}^\infty \frac{1}{m!m!}\left(\frac{3}{2}\right)^{2m}.     
\end{equation*}
Now, for integer $m \geq 2$, we have that $m! \geq \frac{2}{9} \cdot 3^m$
and so,
\begin{equation*}
\frac{1}{m!}\left(\frac{3}{2}\right)^{2m} \leq  \frac{9}{2} \left(\frac{3}{4}\right)^m.
\end{equation*}
Therefore, putting these estimates together, we obtain
\begin{align*}
|R_{M,n}(x)| \leq \frac{18 \cdot (3/4)^{M+1}}{(M+1)!}.
\end{align*}

\subsection{Computation with rational numbers}

First, observe that for any $0<s<\pi$, we may set
$r = \sqrt{1-\frac{s^2}{\pi^2}}\in (0,1)$,
and therefore,
\begin{align*}
(1+r) J_0(\sqrt{1-r^2}\pi)+(1-r) J_2(\sqrt{1-r^2}\pi)
= (1+r) J_0(s)+(1-r) J_2(s).
\end{align*}
If in addition $s\leq 3$, then following the notation and results from the previous section, for any $M \in \NN$, noting that $2 \cdot 18 = 36$,
\begin{align*}
&(1+r) J_0(s) + (1-r) J_2(s)\\
&\,\, \leq (1+r) \sum_{m = 0}^M \frac{(-1)^m}{m!m!}\left(\frac{s}{2}\right)^{2m} + (1-r) \sum_{m = 0}^M \frac{(-1)^m}{m!(m+2)!}\left(\frac{s}{2}\right)^{2m+2} + \frac{36 \cdot (3/4)^{M+1}}{(M+1)!}.
\end{align*}
Assume in addition that $s$ and $M$ are chosen so that 
\begin{equation}\label{eq:final-assumption-before computation}
\begin{split}
r &\geq \frac{2}{5},\\
\sum_{m = 0}^M \frac{(-1)^m}{m!m!}\left(\frac{s}{2}\right)^{2m} &< 0,\\
\sum_{m = 0}^M \frac{(-1)^m}{m!(m+2)!}\left(\frac{s}{2}\right)^{2m+2} &> 0.
\end{split}
\end{equation}
Then, we can say that
\begin{equation}\label{eq:final-bound}
\begin{split}
&(1+r) J_0(s) + (1-r) J_2(s)\\
&\,\, \leq \frac{7}{5}\sum_{m = 0}^M \frac{(-1)^m}{m!m!}\left(\frac{s}{2}\right)^{2m} + \frac{3}{5} \sum_{m = 0}^M \frac{(-1)^m}{m!(m+2)!}\left(\frac{s}{2}\right)^{2m+2} + \frac{36 \cdot (3/4)^{M+1}}{(M+1)!}.
\end{split}    
\end{equation}

With this, we set $s = \frac{287}{100}$ and $M = 5$.
We now verify that the assumptions in \eqref{eq:final-assumption-before computation} are met. Marking equalities with a ``c" when relying on computer output, we have
\begin{align*}
r^2 - \left(\frac{2}{5}\right)^2 = 1-\frac{s^2}{\pi^2} - \left(\frac{2}{5}\right)^2 > 1-\frac{287^2}{314^2}-\left(\frac{2}{5}\right)^2 \overset{\text{c}}{=} \frac{11291}{2464900}.
\end{align*}
Therefore, the first line in \eqref{eq:final-assumption-before computation} is satisfied. For the remaining two lines, we find
\begin{align*}
\sum_{m = 0}^M \frac{(-1)^m}{m!m!}\left(\frac{s}{2}\right)^{2m} &\overset{\text{c}}{=} -\frac{314127831054337257779422849}{1474560000000000000000000000},\\
\sum_{m = 0}^M \frac{(-1)^m}{m!(m+2)!}\left(\frac{s}{2}\right)^{2m+2} &\overset{\text{c}}{=} \frac{170519275716150776952821694135817}{353894400000000000000000000000000},
\end{align*}
which, in summary, shows that all assumptions in \eqref{eq:final-assumption-before computation} are met, and then \eqref{eq:final-bound} is satisfied. In addition,
\begin{equation*}
\begin{split}
&\frac{7}{5}\sum_{m = 0}^M \frac{(-1)^m}{m!m!}\left(\frac{s}{2}\right)^{2m} + \frac{3}{5} \sum_{m = 0}^M \frac{(-1)^m}{m!(m+2)!}\left(\frac{s}{2}\right)^{2m+2} + \frac{36 \cdot (3/4)^{M+1}}{(M+1)!}\\
&\qquad \overset{\text{c}}{=} -\frac{143509674278087403655101304183}{589824000000000000000000000000000},
\end{split}
\end{equation*}
and thus, by \eqref{eq:final-bound},
\begin{equation*}
(1+r) J_0(\sqrt{1-r^2}\pi)+(1-r) J_2(\sqrt{1-r^2}\pi) < 0.   
\end{equation*}
We have therefore proven the assertion made in the beginning of this section.

\section{Proof of \ref{rem:amp-to-flux-free} in Remark \ref{rem:thms}}\label{app:proof-of-amp-to-flux-free}

Let $\cR \subset \RR^3$ be a smooth solid torus of revolution with cross-section $\Sigma$. Analogous to the terminology in Section \ref{sec:sym-amp-intro} with $K$ the rotational Killing field from that section, a flux-free curl eigenfield $B$ on $\cR$ is \emph{symmetric} if $\cL_K B = 0$ and we will call its corresponding eigenvalue a \emph{symmetric flux-free curl eigenvalue}. An analogous result to that of Proposition \ref{prop:amp-sym-is-G-S} holds in this case \cite{EPS23b}.

\begin{proposition}\label{prop:flux-free-sym-is-G-S}
A vector field $B$ in $\cR$ is a symmetric flux-free curl eigenfield if and only if there exists $u \in C^\infty(\Sigma)$ and $c \in \RR$ such that
\begin{align*}
B &= \frac{1}{\lambda}\nabla u(r,z) \times \nabla \varphi + u(r,z) \nabla \varphi,\\
L u &= -\lambda^2 u,\\
u|_{\partial \Sigma} &= c,\\
\int_\Sigma u/r~drdz &= 0,
\end{align*}
where $L$ is the Grad-Shafranov operator and $(r,\varphi,z)$ denote the cylindrical coordinates.
\end{proposition}

Because the flux-free curl spectrum on $\cR$ is discrete (see Appendix \ref{app:func-analysis-antisym-spec}), the set of symmetric flux-free curl eigenvalues on $\cR$ is also discrete. Let $\lambda^{\text{sFF}}_1(\cR)$ and $\lambda^{\text{sA}}_1(\cR)$ respectively denote the first positive symmetric flux-free and Amp\`erian curl eigenvalues. Proposition \ref{prop:flux-free-sym-is-G-S} enables:

\begin{lemma}\label{lem:flux-free-are-bigger-than-amp}
$\lambda^{\text{sFF}}_1(\cR)  > \lambda^{\text{sA}}_1(\cR)$.
\end{lemma}

\begin{proof}
Set $\lambda \equiv \lambda^{\text{sFF}}_1(\cR)$ and let $B$ be a corresponding flux-free eigenfield. In view of Proposition \ref{prop:flux-free-sym-is-G-S}, setting $u_0 := u - c$, with the notation in Appendix \ref{app:Grad-Shafranov},
\begin{equation*}
-(L u_0, u_0)_{L^2,r} = B_r[u_0,u_0].
\end{equation*}
Since $u$ is orthogonal to constants with the inner product $(\cdot,\cdot)_{L^2,r}$,
\begin{equation*}
-(L u_0, u_0)_{L^2,r} =\lambda^2(u_0,u_0)_{L^2,r}.
\end{equation*}
Hence, by Equation \eqref{eq:min-char-of-eigenvalue} in Appendix \ref{app:Grad-Shafranov} and Proposition \ref{prop:amp-sym-is-G-S}, we have that
\begin{equation*}
(\lambda^{\text{sA}}_1(\cR))^2 = \lambda^{\text{G-S}}_1(\Sigma) \leq \frac{B_r[u_0,u_0]}{(u_0,u_0)_{L^2,r}} = \lambda^2.
\end{equation*}
Suppose equality holds. Then $u_0$ would be a first Dirichlet eigenfunction of the Grad-Shafranov operator, and so $c = 0$, yielding that $u_0 = u$. However, the integral condition on $u$ implies that $u$ does not have a sign and hence it cannot be a first Dirichlet eigenfunction for $L$. This contradiction completes the proof of the lemma.
\end{proof}

We are now ready to prove \ref{rem:amp-to-flux-free} in Remark \ref{rem:thms}. Let $\lambda_1^*(\cR)$ denote the first positive antisymmetric curl eigenvalue on $\cR$ (which is a flux free curl eigenvalue according to Lemma \ref{lem:compatibility-of-antisym}). If there do not exist rotationally symmetric eigenfields corresponding to the first (positive) Amp\`erian curl eigenvalue on $\cR$, then
\begin{equation*}
\lambda_1^*(\cR) < \lambda^{\text{sA}}_1(\cR) < \lambda^{\text{sFF}}_1(\cR),
\end{equation*}
where the second inequality is coming from Lemma \ref{lem:flux-free-are-bigger-than-amp}, and we are done. 

\end{document}